 \documentclass[11pt]{article}
\usepackage{latexsym,amssymb,epsfig,bm,amsmath}
\usepackage{color}
\usepackage{amsthm}
\usepackage{mathrsfs}

\newcommand{\beqn}{\begin{eqnarray}}
 \newcommand{\eeqn}{\end{eqnarray}}
 \newcommand{\be}{\begin{equation}}
 \newcommand{\ee}{\end{equation}}
 \newcommand{\ba}{\begin{array}}
 \newcommand{\ea}{\end{array}}
\newcommand{\pa}{\partial}
 \newcommand{\re}{\ref}
 
 \newcommand{\ds}{\displaystyle}
 \newcommand{\la}{\label}
\newcommand{\rIm}{{\rm Im\5}}
 \newcommand{\rRe}{{\rm Re\5}}
\newcommand{\fr}{\frac}

\newcommand{\ov}{\overline}

\newcommand{\supp}{{\rm supp~}}
\newcommand{\bB}{{\bf B}}
\newcommand{\bC}{{\bf C}}
\newcommand{\bP}{{\bf P}}
\newcommand{\bR}{{\bf R}}
\newcommand{\bF}{{\bf F}}
\newcommand{\cO}{{\cal O}}
\newcommand{\cS}{{\cal S}}

\newcommand{\ve}{\varepsilon}

\newcommand{\de}{\delta}

\newcommand{\al}{\alpha}

\newcommand{\om}{\omega}
\newcommand{\rw}{{\rm w}}
\newcommand{\rv}{{\rm v}}
\newcommand{\Om}{\Omega}
\newcommand{\na}{\nabla}

\newcommand{\lam}{\lambda}
\newcommand{\Lam}{\Lambda}
\newcommand{\5}{{\hspace{0.5mm}}}
\newcommand\C{{\mathbb C}}
\newcommand\R{{\mathbb R}}
\newcommand\M{{\mathbb M}}

\renewcommand{\theequation}{\thesection.\arabic{equation}}
\renewcommand{\thesection}{\arabic{section}}
\renewcommand{\thesubsection}{\arabic{section}.\arabic{subsection}}
\newtheorem{theorem}{Theorem}[section]

\renewcommand{\thetheorem}{\arabic{section}.\arabic{theorem}}
\newtheorem{lemma}[theorem]{Lemma}
\newtheorem{remark}[theorem]{Remark}
\newtheorem{cor}[theorem]{Corollary}
\newtheorem{pro}[theorem]{Proposition}
\newcommand{\bp}{\begin{proof}}
\newcommand{\ep}{\end{proof}}
\newcommand{\bc}{\begin{cor}}
\newcommand{\ec}{\end{cor}}
\newcommand{\br}{\begin{remark} }
\newcommand{\er}{\end{remark}}
\newcommand{\const}{\mathop{\rm const}\nolimits}

\begin{document}

\begin{titlepage}

\bigskip\bigskip\bigskip

\begin{center}
{\Large\bf
On asymptotic stability of solitons
\bigskip\\
for nonlinear Schr\"odinger equation}
\vspace{1cm}

{\large A.~I.~Komech}
\footnote{
Supported partly by Alexander von Humboldt Research Award,
by FWF and RFBR grants.
}$^{\!,\5 2}$\\
{\it Fakult\"at f\"ur Mathematik, Universit\"at Wien\\
and Institute for the Information Transmission  Problems  RAS}\\
 e-mail:~alexander.komech@univie.ac.at
\medskip\\
{\large E.~A.~Kopylova}
\footnote{Supported partly by FWF, DFG and
RFBR grants.}\\
{\it Institute for the Information Transmission  Problems  RAS\\
B.Karetnyi per.19.,Moscow 101447, Russia}\\
  e-mail:~ek@vpti.vladimir.ru
\medskip\\
{\large D.~Stuart}
\footnote{Partially supported by EPSRC grant PP/D507366/1}\\
{\it Centre for Mathematical Sciences,\\
Wilberforce Road, Cambridge, CB3 OWA}\\
 e-mail:~D.M.A.Stuart@damtp.cam.ac.uk
\end{center}

\date{}

\markboth{A.Komech, E.Kopylova}
  {On Asymptotic Stability of Solitary Waves
\bigskip\\
for Schr\"odinger Equation coupled
to nonlinear oscillator, II}

\vspace{0.5cm}
\begin{abstract}

 The long-time asymptotics is analyzed for 
finite energy solutions of the 1D Schr\"odinger equation coupled to a 
nonlinear oscillator; mathematically the system under study is a nonlinear 
Schr\"odinger equation, whose nonlinear term has spatial dependence of a 
Dirac delta function. The coupled system is invariant with respect to the 
phase rotation group $U(1)$. This article, which extends the results of a 
previous one, provides a proof of asymptotic stability of solitary wave 
solutions in the case that the linearization contains a single discrete 
oscillatory mode satisfying a non-degeneracy assumption of the type known 
as the Fermi Golden Rule. This latter condition is proved to hold 
generically. 
\end{abstract}

\end{titlepage}
\setcounter{equation}{0}
\section{Introduction and statement of results}
\subsection{Introduction}
\label{intr}
In this article we continue the study, initiated in \cite{BKKS}, of large 
time asymptotics for a model $U(1)$-invariant nonlinear Schr\"odinger 
equation 
\be\la{S}
  i\dot\psi(x,t)=
  -\psi''(x,t)-\de(x)F(\psi(0,t)),\quad x\in\R,
\ee 
Here $\psi(x,t)$ is a
continuous complex-valued wave function and $F$ is a continuous
function. Our main focus is on the role that certain solitary
waves (also referred to as nonlinear bound states, or solitons)
play in the description of the solution for large times. These
solitary waves are solutions of the form $e^{i\om t}\psi_\om(x)$,
where $\psi_\om$ solves a nonlinear eigenvalue problem
\eqref{NEP}. In \cite{BKKS} the asymptotic stability of these
solitary waves was proved under a condition on the nonlinearity
which ensures that the linearization about the solitary wave
consists entirely of continuous spectrum, except for the two
dimensional generalized null space which is always present due to
the $U(1)$ symmetry of the equation. In this article this result
is extended to the case that the spectrum of the linearization
includes an additional discrete component, which satisfies a
non-degeneracy condition related to the Fermi Golden Rule. In
order to explain these results fully we will introduce our
conditions on the nonlinearity $F$, in the remainder of this
section. In the following section we will discuss the basic
properties of the solitary waves. We will then be able to state
our main theorem precisely as theorem \ref{main}. For a more
lengthy discussion of our motivation, and of previous results in
the literature (\cite{BS,PW92,PW94,PW,SW1,SW2}) we refer the
reader to the introduction of \cite{BKKS}.

Let us emphasize that we do not hypothesize any spectral properties of the
linearized equation, but derive explicitly everything that is required. 
This is posibble on account of the simplicity of our
model which allows an exact analysis all spectral properties.
Thus we are able to give a complete proof, without any suppositions, of the 
soliton asymptotics (\ref{sol-as'}), and also of the 
generic validity of the Fermi Golden rule (\ref{FGR}).

It is convenient to rewrite \eqref{S} in real form:
we identify a complex number $\psi=\psi_1+i\psi_2$ with the real two-dimensional
vector $(\psi_1,\psi_2)\in\R^2$ and rewrite (\re{S}) in the vectorial form
\be\la{SV}
  j\dot\psi(x,t)= -\psi''(x,t)-\de(x)\bF(\psi(0,t)),\quad
  j=\left(\ba{rr}
   0  &  -1\\
   1  &   0
   \ea\right),
\ee
where $\bF(\psi)\in\R^2$  is the real vector version of $F(\psi)\in\C$.
We assume that the oscillator force $\bf F$ admits a real-valued potential
\be\la{bF}
  {\bf F}(\psi)=-\na U(\psi),\quad\psi\in\R^2,\quad U\in C^2(\R^2).
\ee
Then (\re{SV}) is formally a Hamiltonian system with  Hamiltonian
\be\la{H}
 {\cal H}(\psi)=\fr 12\int  |\psi'|^2 dx+U(\psi(0))
\ee
We assume that the potential $U(\psi)$ satisfies the
inequality \be\la{U}
   U(z)\ge A-B|z|^2 \quad {\rm with\; some}\quad A\in\R,\quad B>0.
\ee
We also assume that $U(\psi)=u(|\psi|^2)$ with $u\in C^2(\R)$.
Therefore, by (\re{bF}),
\be\la{I}
   F(\psi)= a(|\psi|^2)\psi,\quad\psi\in\C\5,\quad a\in C^1(\R),
\ee
where $a(|\psi|^2)$ is real. Then
$F(e^{i\theta}\psi)=e^{i\theta} F(\psi)$, $\theta\in[0,2\pi]$,
and $e^{i\theta}\psi(x,t)$ is a solution to (\re{S})  if $\psi(x,t)$ is.
Therefore, equation (\ref{S}) is $U(1)$-invariant and the N\"other theorem implies
the conservation of the following {\em charge}:
$$
  {\cal Q}(\psi)=\int |\psi|^2 dx=\const.
$$

Under these conditions the existence of global solutions to the Cauchy problem
for \eqref{S} was proved in \cite{KoK06}. Let  $C_b(\R,X)$ be
the space of bounded continuous functions $\R\to X$ into a Banach
space $X$.
\begin{theorem}[\cite{KoK06}]\label{locex}
Under conditions (\re{bF}), (\re{U}) and  (\re{I}), the following statements hold.\\
i)\;\;\;  For any $\psi_0(x)\in H^1(\R)$ there exists a unique solution
   $\psi(t)=\psi(\,\cdot\,,t)\in C_b(\R,H^1(\R))$
   to the equation (\ref{S}) with initial condition
$\psi(x,0)=\psi_0(x)$.\\
ii)\;\;  The charge ${\cal Q}(\psi(t))$ and
Hamiltonian ${\cal H}(\psi(t))$
are constant along the solution.\\
iii)\;  There exists $\Lambda(\psi_0)>0$ such that the following a priori bound holds:
   \be\label{a-priori-h1}
     \sup\limits\sb{t\in\R}\Vert{\psi(t)}\Vert\sb{H^1(\R)}\le \Lam(\psi_0)<\infty.
   \ee
\end{theorem}

In \S\ref{swsec} we describe all
nonzero solitary waves, and then formulate the main theorem
in \S\ref{mainstat}.
\subsection{Solitary waves}
\label{swsec}

Equation (\re{S}) admits finite energy  solutions of type $\psi_\om(x)e^{i\om t}$, called
{\it solitary waves} or {\it nonlinear eigenfunctions}.
The frequency  $\om$ and the amplitude $\psi_\om(x)$ solve the following
{\it nonlinear eigenvalue problem}:
\be\la{NEP}
  -\om\psi_\om(x)= -\psi_\om''(x)-\de(x)F(\psi_\om(0)),\quad x\in\R.
\ee
It is straightforward to check (see \cite{BKKS}) that the set of all
nonzero solitary waves consists of functions
$\psi_\om(x)e^{i\theta}$, $\psi_\om(x)=C(\om)e^{-\sqrt\om |x|}$, $C>0,\,\om>0$, where
$\sqrt\om =a(C^2)/2>0,$
and  $\theta\in [0,2\pi]$ is arbitrary.
This condition means that $C$ is restricted to lie in a set which, in the case of
polynomial $F$, is a finite union of one-dimensional intervals.
Notice that $C=0$ corresponds to the zero function $\psi(x)=0$ which
is always a solitary  wave  as $F(0)=0$, and for $\om\le 0$ only the zero solitary wave exists.
The real form of the solitary wave is $e^{j\theta}\Psi_\om$ where $\Psi_\om=(\psi_{\omega}(x),0)$.
We will also need the following lemma from \cite{BKKS}:
\begin{lemma}\la{int-dif}
For $C> 0$, $a>0$ we have
$$
  \pa_\om\int|\psi_\om(x)|^2 dx>0\quad\hbox{if}\;\; 0<a'<a/C^2.\la{syc2}
$$
\end{lemma}
Linearization at the solitary wave $e^{j\theta}\Psi_\om$ leads to the operator
\be\la{B}
{\bf B}=-\ds\fr{d^2}{dx^2}+\om-\de(x)[a(C^2)+2a'(C^2)C^2P_1]=
\left(
   \ba{cc}
   {\bf D}_1  &       0    \\
        0     &   {\bf D}_2
   \ea
   \right),
\ee
where $P_1$ is the projector in $\R^2$ acting as
$(\chi_1,\chi_2)\mapsto (\chi_1, 0)$, and (see \cite{BKKS})
$$
   {\bf D}_1=-\ds\fr{d^2}{dx^2}+\om-\de(x)[a+2a'C^2],\quad
   {\bf D}_2=-\ds\fr{d^2}{dx^2}+\om-\de(x)a.
$$
Let $\bC=j^{-1}\bB$. The spectrum and eigenfunctions of $\bC$ are
computed explicitly in  \cite{BKKS}, and the results we need are
summarized in the \S\ref{subspace} and the appendices. Briefly,
\begin{itemize}
\item
the continuous spectrum coincides with
${\mathcal C}_+\cup{\mathcal C}_-$ where
${\mathcal C}_+=[i\om, i\infty)$, and
${\mathcal C}_-=(-i\infty,-i\om]$;
\item
the discrete spectrum always contains zero
on account of the circular symmetry of the problem, and there is a
corresponding generalized eigenspace of dimension at least two.
If $a'>a/C^2$ there is a positive eigenvalue corresponding to linear
instability of the solitary wave, while for $a'<a/C^2$ the discrete
spectrum consists either only of zero, or contains in addition two
pure imaginary eigenvalues.
\end{itemize}
\subsection{Statement of  main theorem}\la{mainstat}
Previously, in \cite{BKKS}, we considered the case when
$a'\in (-\infty,0)\cup(0,a/\sqrt 2C^2)$.
In which case the operator $\bC$ has no discrete spectrum except zero.
In the present paper we will consider the case when
\be\la{Co}
  a'\in\big(a/\sqrt 2 C^2,~a\sqrt 2(1+\sqrt 3)/4C^2\big).
\ee
In this case, there are, in addition to zero,
2 simple pure imaginary eigenvalues $\pm i\mu$,
which satisfy the property $2\mu>\om$.
If assumption (\ref{Co}) is true for a fixed value $\om_0$, it also true
for values of $\om$ in a small interval centered at $\om_0$.
Let $u(x,\om)=(u_1,u_2)$ be the eigenvector of $\bC$ associated
to $i\mu$. We choose the function $u_1(x)$ to be real,
in which case $u_2(x)$ is purely imaginary.
Then  $u^*:=(u_1,-u_2),$ is the
eigenvector  associated to $-i\mu$ (see appendix A).
We consider the initial value $\psi(x,0)=\psi_0(x)$
to be of the form
\be\la{id-form}
  \psi_0(x)=e^{j\theta_0}
\Big[\Psi_{\om_0}(x)+z_0u(x,\om_0)+\ov z_0u^*(x,\om_0)+f_0(x)\Big],
\ee
where  $f_0$ belongs to the eigenspace associated
to the continuous spectrum of $\bC(\om_0)$. We assume that
$z_0$ and $f_0$ are
sufficiently small, and also assume a non-degeneracy condition
which we now explain.
Let $\langle\cdot,\cdot\rangle$ denote the Hermitian scalar product
in $L^2$ of $\C^2$-valued
functions, and $(u,v)=u_1v_1+u_2v_2$ for $u,v\in\C^2$.
Let $E_2[f,f]$ be the quadratic terms coming from the Taylor
expansion of the nonlinearity:
\be\la{e2d}
  E_2[f,f]=\delta(x)[a'(C^2)(f,f)\Psi_{\om_0}+2a''(C^2)(\Psi_{\om_0},f)^2\Psi_{\om_0}
  +2a'(C^2)(\Psi_{\om_0},f)f],\;f\in\C^2.
\ee
The non-degeneracy condition has the form
\be\la{FGR}
  \langle E_2[u(\om_0),u(\om_0)],\tau_+(2i\mu_0)\rangle\not =0,
\ee
where $\tau_+(2i\mu_0)$ is the eigenfunction associated to $2i\mu_0=2i\mu(\om_0)$.
This condition should be thought of as a nonlinear version of the
Fermi Golden Rule of quantum mechanics \cite{MS}.
In appendix E we express (\ref{FGR}) in terms of $C$ and $a(C^2)$,
and hence show that 
{\it the Fermi Golden Rule holds generically for polynomial nonlinearity.}\\
Let us introduce the weighted Banach space $L^p_{\beta}$ with the finite norm
$$
  \Vert f\Vert_{L^p_{\beta}}=\Vert (1+|x|)^{\beta} f(x)\Vert_{L^p}
$$
Our main  theorem is the following:
\begin{theorem}\label{main}
   Let conditions (\re{bF}), (\re{U})  and  (\re{I}) hold,
   $\beta> 2$ and $\psi(x,t)\in C(\R,H^1)$ be the solution to the equation (\ref{SV})
   with initial value $\psi_0(x)=\psi(x,0)\in H^1\cap L^1_\beta$ of the form (\ref{id-form})
   which is close to a solitary wave $e^{j\theta_0}\Psi_{\om_0}$:
   \be\la{close}
    |z(0)|\le\ve^{1/2},\quad \Vert f_0\Vert_{L^1_{\beta}}\le c\ve^{3/2}.
   \ee
   Assume further that the spectral condition (\ref{Co}) and the non-degeneracy
   condition (\ref{FGR}) hold for the solitary wave with $C=C(\om_0)=C_0$. Then
   for $\ve$ sufficiently small the solution admits the following scattering
   asymptotics in $C_b(\R)\cap L^2(\R)$:
   \be\la{sol-as}
    \psi(x,t)=e^{j\varphi_{\pm}(t)}\Psi_{\om_{\pm}}(x)
   +e^{j^{-1}Lt}\Phi_{\pm}+ O(t^{-\nu}),
   \quad t\to\pm\infty,
   \ee
  with any $0<\nu<1/4$, where  $L=-\frac{\partial^2}{\partial x^2}$,
  $\Phi_{\pm}\in C_b(\R)\cap L^2(\R)$ are the corresponding asymptotic scattering states and
  $\varphi_{\pm}(t)=\om_{\pm}t+p_{\pm}\log(1+k_{\pm}t)+\varkappa_{\pm}$,
  for some constants $\om_{\pm}$, $p_{\pm}$, $k_{\pm}$, $\varkappa_{\pm}$.
\end{theorem}
The asymptotics (\ref{sol-as}) can be rewritten in terms of the
original complex notation as:
\be\la{sol-as'}
 \psi(x,t)=e^{i\varphi_{\pm}(t)}\psi_{\om_{\pm}}(x)
+W(t)\phi_{\pm}+ O(t^{-\nu}),
  \quad t\to\pm\infty,
\ee
where $W(t)$ is the dynamical group of the free Schr\"odinger equation,
and $\phi_{\pm}=\Phi_\pm^1+i\Phi_\pm^2$
($\Phi_\pm^k$, $k=1,2$, being the components of the
vector-function $\Phi_{\pm}$).
\setcounter{equation}{0}
\section{Linearization}
In this section we summarize the spectral properties of the
operator $\mathbf{C}$ and then give some estimates for the
linearized evolution.
The proof of these properties can be found in \cite{BKKS}, with the
exception of proposition \ref{l1} which is proved in appendix C.
\subsection{Spectral properties}
\label{subspace}
The linearized equation reads
\be\la{lin4}
   \dot\chi(x,t)={\bf C}\chi(x,t),~~~~~{\bf C}:=j^{-1}\bB=
   \left(
   \ba{rr}
    0          &    {\bf D}_2\\
   -{\bf D}_1  &        0
   \ea
   \right).
\ee
Theorem \ref{locex} generalizes to the equation (\ref{lin4}):
the equation admits unique solution $\chi(x,t)\in C_b(\R,H^1)$
for every initial function $\chi(x,0)=\chi_0\in H^1$.
Denote by $e^{{\bf C}t}$ the dynamical group of equation (\ref{lin4})
acting in the space $H^1$; for $T>0$ there exists $c_T>0$
such that
\be\la{t-small}
  \Vert e^{{\bf C}t}\chi_0\Vert_{H^1}\le c_T\Vert\chi_0\Vert_{H^1},\qquad |t|\leq T.
\ee
The resolvent $\bR(\lam):=(\bC-\lam)^{-1}$ is an  integral operator
with matrix valued integral kernel
\be\la{nR}
  {\bf R}(\lam,x,y)=\Gamma(\lam,x,y)+P(\lam,x,y),
\ee
where
\beqn\nonumber
   \Gamma(\lam,x,y)&=&\left( \ba{cc}
   \ds\frac 1{4k_+}  &  -\ds\frac 1{4k_-}\\
   \ds\frac i{4k_+}  &  \ds\frac i{4k_-}
   \ea\right)\left( \ba{cc}
   e^{ik_+|x-y|}-e^{ik_+(|x|+|y|)}  &  -i(e^{ik_+|x-y|}-e^{ik_+(|x|+|y|)})\\\\
   e^{ik_-|x-y|}-e^{ik_-(|x|+|y|)}  &   i(e^{ik_-|x-y|}-e^{ik_-(|x|+|y|)})
   \ea\right)\\
\nonumber
   P(\lam,x,y)&=&\frac 1{2D} \left( \ba{cc}
   e^{ik_+|x|}    &   e^{ik_-|x|}\\
   ie^{ik_+|x|}   &  -ie^{ik_-|x|}
   \ea\right)
   \left( \ba{cc}
   i\al-2k_-       &    i\beta\\
   -i\beta         &   -i\al+2k_+
   \ea\right)
   \left( \ba{cc}
   e^{ik_+|y|}      &  -ie^{ik_+|y|}      \\
   e^{ik_-|y|}       &  ie^{ik_-|y|}
   \ea\right).
\eeqn
The notation is as follows: as explained 
already in \S\ref{swsec}, the continuous spectrum consists of ${\cal 
C}_+\cup{\cal C}_-$, and correspondingly $k_\pm(\lam)=\sqrt{-\om\mp i\lam}$ 
is (respectively) the square root defined with a cut in the complex $\lam$ 
plane so that $k_\pm(\lambda)$ is analytic on ${\C}\setminus {\cal C}_\pm$ 
and ${\rm Im}\5 k_\pm(\lambda)>0 $ for $\lam\in\C\setminus {\cal C}_\pm$. 
The constants $\al$, $\beta$ and $D=D(\lam)$ are given by the formulas
$$
 \al=a+a'C^2,\;\beta=a'C^2,\;D=2i\al(k_++k_-)-4k_+k_-+\al^2-\beta^2. 
$$ 
In addition to this continuous spectrum, there is discrete spectrum, which
appears in this formalism as the set of
poles of $\mathbf{R}(\lambda)$; these poles in turn
correspond to the roots of the determinant $D(\lam)$. From the analysis
in \cite{BKKS} we know that
if $a'\in (a/{\sqrt 2 C^2},a/C^2)$,
then the determinant has the following roots:
$\lam =0$ with  multiplicity $2$ and two pure imaginary roots
\be\la{gaga}
\lam=\pm i\mu=\pm i\frac{\beta}2\sqrt{a^2-\beta^2},\quad\mu<\om.
\ee
with  corresponding eigenfunctions $u$ and $u^*$ which are displayed in 
\eqref{u-def} in the appendix. Note that the spectral condition (\ref{Co}) 
is more restrictive: it implies in addition that $2\mu>\om$, as can be 
verified by a simple computation.

The generalised null space $X^0$ of the non-self-adjoint operator ${\bf C}$ is two dimensional,
is spanned by $j\Psi_\om,\partial_\om\Psi_\om$, and
$$
{\bf C} j\Psi_\om=0,\qquad {\bf C} \partial_\om\Psi_\om=j\Psi_\om.
$$
The symplectic form $\Omega$ for the vectors $\psi$ and $\eta$ is defined by
\be\la{symp}
   \Om(\psi,\eta)=\langle\psi,j\eta\rangle
\ee
By Lemma \ref{int-dif}
\be\label{nondeg}
   \Omega(j\Psi_\om,\partial_\om\Psi_\om)
   =-\frac{1}{2}\partial_\om\int |\psi_\om|^2dx\not=0.
\ee
Hence, the symplectic form $\Omega$ is nondegenerate on $X^0$,
i.e. $X^0$ is a symplectic subspace. There exists a symplectic projection operator
${\bf P}^0$ from $L^2({\R})$ onto $X^0$ represented by the formula
\be\label{defsp}
  {\bf P}^0\psi=\frac 1{\langle\Psi_\om,\partial_\om\Psi_\om\rangle}
  [\langle\psi,j\partial_\om\Psi_\om\rangle j\Psi_\om
  +\langle \psi,\Psi_\om\rangle\partial_\om\Psi_\om].
\ee
\begin{remark}
  Since $j\Psi_\om, \partial_\om\Psi_\om$ lie in $H^1({\R})$ the operator $\bP^0$
  extends uniquely to define a continuous linear map $H^{-1}({\R})\to X^0$.
  In particular this operator can be applied to the Dirac measure $\delta(x)$.
\end{remark}
Denote by $X^1$ the eigensubspace corresponding to the two pure imaginary 
eigenvalues, and by $\bP^1$ a symplectic projection operator from $L^2({\R})$ onto $X^1$.
It may be represented by the formula
\be\label{defspd}
  \bP^1\psi=\frac{\langle\psi,ju\rangle}{\langle u,ju\rangle}u
  +\frac{\langle\psi,ju^*\rangle}{\langle u^*,ju^*\rangle}u^*
\ee
since
$\langle u,ju^*\rangle=0$, and
 $\langle u,ju\rangle=\ov{\langle u^*,ju^*\rangle}\not=0$ by (\ref{uu}).
Finally, $\bP^c=1-\bP^0-\bP^1$ is the symplectic projector onto
the continuous spectral subspace.
\subsection{Estimates for  linearized evolution}
\label{rpsec}
We now recall from \cite{BKKS} some estimates on the  group  $e^{{\bf C}t}$
which will be needed in \S\ref{maj}.
First we recall a bound for the action of $e^{\bC t}$ on the Dirac distribution
$\de=\de(x)$ for small $t$.  Lemma 8.1 from \cite{BKKS} gives the following
small $t$ behaviour:
\be\la{delta1}
  \Vert e^{{\bf C}t}\delta\Vert_{L^{\infty}}={\cO}(t^{-1/2}),~~~~~~~t\to 0.
\ee
Second we list the large time dispersive estimates. To do this
let us introduce,
for $\beta\ge 2$, a Banach space ${\cal M}_\beta$, which is the subset
of distributions which are  linear combinations of $L^1_\beta$ functions and
multiples of the Dirac distribution at the origin with the norm:
\be\la{cM}
  \Vert \psi+C\delta(x)\Vert_{{\cal M}_\beta}:=\Vert\psi\Vert_{L^1_\beta}+|C|.
\ee
\begin{pro}\la{TD}(see \cite{BKKS})
Assume that the spectral condition (\ref{Co}) holds. Then for
$h\in{\cal M}_{\beta}$ with $\beta\ge2$  the following bounds hold:
\be\la{b1}
    \Vert e^{{\bf C}t}{\bf P}^ch\Vert_{L^{\infty}_{-\beta}}
    +\Vert e^{{\bf C}t}{\bf\Pi}^{\pm}h\Vert_{L^{\infty}_{-\beta}}
  +\Vert e^{{\bf C}t}{\bf C}^{-1}{\bf\Pi}^{\pm}h\Vert_{L^{\infty}_{-\beta}}
    \le c (1+t)^{-3/2}\Vert h\Vert_{{\cal M}_{\beta}},
\ee
where ${\bf\Pi}^{+}$ (resp. ${\bf\Pi}^{-}$) is
the spectral projection operator
onto the spectral subspace associated to ${\cal C}_+$ (resp. ${\cal C}_-$),
the positive (resp. negative) part of the
continuous spectrum.
\end{pro}
We shall also need the following bound, which is new.
\begin{pro}\label{l1}
Assume that the spectral condition (\ref{Co}) holds. Then
for $h\in{\cal M}_{\beta}$ with $\beta\ge 2$  the following bounds hold:
  \be\la{b5}
    \Vert e^{{\bf C}t}({\bf C}\mp 2i\mu-0)^{-1}{\bf P}^ch\Vert_{L^{\infty}_{-\beta}}
    +\Vert e^{{\bf C}t}({\bf C}\mp 2i\mu-0)^{-1}{\bf\Pi}^{\pm}h\Vert_{L^{\infty}_{-\beta}}
    \le c(1+t)^{-3/2}\Vert h\Vert_{{\cal M}_{\beta}}.
  \ee
\end{pro}
We prove this proposition in appendix C.
\setcounter{equation}{0}
\section{Modulation equations }
\label{modsec}
Here we present the modulation equations which allow a
construction of solutions $\psi(x,t)$ of the equation (\re{S})
close at each time $t$ to a soliton i.e. to one of the functions
$Ce^{i\theta-\sqrt{\om}|x|}$
in the set $\cS$ described in \S\ref{swsec} with time varying
(``modulating'') parameters $(\omega,\theta)=(\omega(t),\theta(t))$.

We look for a solution to (\ref{SV}) in the form
\be\la{sol}
  \psi(x,t)=
  e^{j\theta(t)}\Phi(x,t),\quad \Phi(x,t)=\Psi_{\om(t)}(x)+\chi(x,t).
\ee
Since this is a solution of (\ref{SV}) as long as
$\chi\equiv 0$ with $\dot\theta=\om$ and $\dot \omega=0$
it is natural to look for solutions  in which $\chi$ is small and
$\theta(t)=\ds\int_0^t\om(s)ds+\gamma(t)$ with $\gamma$ treated
perturbatively.

We look for $\chi=\rw(x,t)+f(x,t)$ where $\rw=zu+{\overline z}u^*\in X^1$
and $f\in X^c$.
Now we give a system of coupled equations
for $\omega(t)$, $\gamma(t)$, $z(t)$ and $f(x,t)$.
\begin{lemma}\label{BS}(cf.\cite[Proposition 2.2]{BS})
Given a solution of \eqref{SV} in the form \eqref{sol} with $f\in X^c$
as just described, the functions
$\omega(t)$, $\gamma(t)$, $z(t)$ and $f(x,t)$ satisfy the system
  \be\la{omega}
    \dot\omega=\frac{\langle {\bf P}^0{\bf Q}[\chi],\Phi\rangle}
    {\langle\pa_\om\Psi_{\om}-\pa_\om {\bf P}^0\chi,\Phi\rangle},
  \ee
  \be\la{gamma}
    \dot\gamma=\frac{\langle {\bf Q}[\chi], j(\pa_\om\Psi_{\om}-\pa_\om {\bf P}^0\chi)\rangle}
    {\langle\pa_\om \Psi_\om-\pa_\om {\bf P}^0\chi,\Phi\rangle},
  \ee
  \be\la{z}
    \langle u,ju\rangle(\dot z-i\mu z)=\langle {\bf Q}[\chi],ju\rangle-
    \langle\pa_\om \rw-\pa_\om{\bf P}^1f,ju\rangle\dot\om-
    \langle\chi,u\rangle\dot\gamma
  \ee
  \be\la{f}
    \dot f={\bf C}f+\dot\om\pa_\om{\bf P}^c\chi-\dot\gamma{\bf P}^c(j\chi)
    +{\bf P}^c{\bf Q}[\chi],
  \ee
where
${\bf Q}[\chi]=-\delta(x) j^{-1}\bigl({\bf F}(\Psi_\om+\chi)
-{\bf F}(\Psi_\om)-{\bf F}'(\Psi_\om)\chi\bigr)$ represents the nonlinear
part of the interaction.
\end{lemma}
\setcounter{equation}{0}
\section{Frozen spectral decomposition }
\label{f-dec}
The linear part of the evolution equation (\ref{f}) for $f$ is non-autonomous,
due to the dependence of the operator ${\bf C}$ on $\omega(t)$. In order to
make use of the dispersive properties obtained in \S\ref{rpsec},
we  introduce (following \cite{BS})  a small
modification of (\ref{f}), which leads to an autonomous  equation.
We fix an  interval $[0,T]$ and decompose $f(t)\in X_{t}^c$ into the sum
\be\la{gh}
f=g+h, \quad g\in X^d_{T}, \quad h\in X^c_{T}.
\ee
Here $X_{T}^d={\bf P}_{T}^dX$ is the spectral
space associated to the discrete spectrum at time $T$ and $X_{T}^c={\bf P}_{T}^cX$
is the spectral space associated to the continuous spectrum at time $T$,
${\bf P}_{T}^c={\bf P}^c(\omega(T))$ and ${\bf P}_{T}^d=I-{\bf P}_{T}^c$.
In the following, we denote $\omega_{T}=\omega(T)$ and ${\bf C}_{T}={\bf C}(\omega_{T})$.
We will obtain estimates uniform in $T$, and later consider
the limit $T\to+\infty$.
\\
We  introduce a shorthand for the bounds we are about to prove:
${\cal R}(A,\!B,\!...)$ (resp. ${\cal R}(\om,A,\!...)$) is a general notation for
a positive function which remain bounded as $A,B,\dots$ approach zero
(resp. if $\om$ is close to $\om_0$ and $A,\dots$ approach zero).
Denote ${\cal R}_1(\om)\!=\!{\cal R}(\Vert\om-\om_0\Vert_{C[0,T]})$.
\begin{lemma}\label{g-est}
  The function $g$ is estimated in terms of $h$ as follows:
   \be\la{g}
     \Vert g\Vert_{L^{\infty}_{-\beta}}={\cal R}_1(\omega)|\omega-\omega_{T}|
     \Vert h\Vert_{L^{\infty}_{-\beta}}
  \ee
\end{lemma}
\bp
  Using the identities
  ${\bf P}^d(g+h)=0$, ${\bf P}_{T}^dg=g$ and ${\bf P}_{T}^dh=0$, we get
  $$
    g+({\bf P}^d-{\bf P}_{T}^d)g+({\bf P}^d-{\bf P}_{T}^d)h=0,
  $$
  and ${\bf P}^d-{\bf P}_{T}^d$ is a``small'' finite dimensional operator:
  $$
    |{\bf P}^d-{\bf P}^d_T|=|{\bf P}^d(\omega_t)-{\bf P}^d(\omega_T)|
    \le {\rm max}_ {\omega^{*}\in(\omega, \omega_{T})}
    |\pa_\om{\bf P}^d(\omega^{*})||\omega-\omega_{T}|.
  $$
\ep

Applying the projection ${\bf P}_{T}^c$ to  (\ref{f}), we obtain
\be\la{h}
  \dot h={\bf C}_{T}h+{\bf P}_{T}^c[({\bf C}-{\bf C}_{T})f
  +{\bf P}^c{\bf Q}[\chi]+\dot\om\pa_\om{\bf P}^c\chi-\dot\gamma{\bf P}^c(j\chi)].
\ee

\setcounter{equation}{0}
\section{Taylor expansion of dynamics}
\label{leadterms}
The preceding sections have provided a change of variables
$\psi\mapsto (\omega,\gamma,z,h)$ under which \eqref{SV}
is mapped into the system comprising \eqref{omega}-\eqref{z} and \eqref{h}.
Since we are interested in proving that for large
times $z,h$ are small it is necessary to expand the inhomogeneous terms
in these equations in terms of $z,h$. This is carried out in this section.
\subsection{Preliminaries}
\label{prelim}
This section is devoted to some useful preliminary estimates.
We start with a bound for the denominator
$\langle\pa_\om \Psi-\pa_\om {\bf P}^0\chi,\Phi\rangle$, where $\Psi=\Psi_\om$,
that appears in the equation of motion (\ref{omega})-(\ref{gamma}).
We have, with $\Delta=  \langle\pa_\om \Psi,\Psi\rangle$,
\be\la{den1}
  \langle\pa_\om \Psi\!-\!\pa_\om {\bf P}^0\chi,\Phi\rangle\!=\!
  \langle\pa_\om \Psi,\Psi\rangle\Bigl(1\!+\!\frac{\langle\pa_\om \Psi,\chi\rangle
  \!-\!\langle\pa_\om {\bf P}^0\chi,\Phi\rangle}{\langle\pa_\om \Psi,\Psi\rangle}\Bigr)
  \!=\!\Delta \Bigl(1\!+\!\frac{\langle\pa_\om \Psi,\chi\rangle
  \!-\!\langle\pa_\om {\bf P}^0\chi,\Phi\rangle}{\Delta}\Bigr)
\ee
with
\be\la{den2}
  \frac{\langle\pa_\om \Psi,\chi\rangle-\langle\pa_\om {\bf P}^0\chi,\Phi\rangle}
  {\Delta}
  ={\cal R}(\om)\Bigl(|z|+\Vert f\Vert_{L^{\infty}_{-\beta}}
  +\Vert f\Vert_{L^{\infty}_{-\beta}}^2\Bigr).
\ee
We also need to expand the nonlinear term ${\bf F}(\psi)=a(|\psi|^2)\psi$
near the solitary wave since the inhomogeneous terms all involve $E[\chi]$,
the nonlinear part of $\delta(x) F$, defined using
the Taylor expansion of $\delta(x){\bf F}\psi$ near $\Psi$:
\be\la{den4}
  \delta(x){\bf F}(\psi)=
  \delta(x)\Bigl(a(C^2)\Psi+a(C^2)\chi+2a'(C^2)(\chi,\Psi)\Psi\Bigr)+E[\chi].
\ee
Thus $E[\chi]$ contains all the higher order terms which are
at least quadratic in $\chi$, as $\chi\to 0$.
We expand $E[\chi]$ in the form
\be\la{E}
E[\chi]=E_2+E_3+E_R,
\ee
where $E_j$ is of order $j$ in $\chi$ and $E_R$ is the remainder.
It is easy to check that
$$
  E_2[\chi,\chi]=\delta(x)\Bigl[a'(C^2)|\chi|^2\Psi+2a''(C^2)(\Psi,\chi)^2\Psi
  +2a'(C^2)(\Psi,\chi)\chi\Bigr],
$$
$$
  \!\!\!\! E_3[\chi,\chi,\chi]\!=\!\delta(x)\Bigl[a'(C^2)|\chi|^2\chi
  +2a''(C^2)(\Psi,\chi)^2\chi
  +2a''(C^2)(\Psi,\chi)|\chi|^2\Psi+\frac 43 a'''(C^2)(\Psi,\chi)^3\Psi\Bigr].
$$
For $E_R$ we have
\be\la{ER}
   E_R={\cal R}(\om,|z|,|f(0)|)(|z|^4+|f(0)|^4).
\ee
It also useful to define $E_2[\chi_1,\chi_2]$,
(resp. $E_3[\chi_1,\chi_2,\chi_3]$)
as a symmetric bilinear (resp. trilinear) form
$$
  E_2[\chi_1,\chi_2]=\delta(x)\Bigl[a'(C^2)(\chi_1,\chi_2)\Psi
  +2a''(C^2)(\Psi,\chi_1)(\Psi,\chi_2)\Psi
  +a'(C^2)\Bigl((\Psi,\chi_2)\chi_1+(\Psi,\chi_1)\chi_2\Bigr)\Bigr],
$$
\beqn\nonumber
E_3[\chi_1,\chi_2,\chi_3]&=&\delta(x)\Bigl[
\fr 16a'(|\Psi|^2)\sum (\chi_i,\chi_j)\chi_k
+\fr 13 a''(|\Psi|^2)\sum (\Psi,\chi_i)(\Psi,\chi_j)\chi_k\\
\nonumber
&+&\fr 13 a''(|\Psi|^2)\sum (\Psi,\chi_i)(\chi_j,\chi_k)\Psi+
\fr 43  a'''(|\Psi|^2)(\Psi,\chi_1)(\Psi,\chi_2)(\Psi,\chi_3)\Psi\Bigr].
\eeqn
Here summation is taken over all permutations of integers $1,2,3$.
Notice also that
\be\la{XYZ}
  \langle E_2[X,Y],Z\rangle=\langle X, E_2[Y^*,Z]\rangle
\ee
where $X$, $Y$, $Z$, are complex valued vector functions and
$Y^*=(\overline Y_1,\overline Y_2)$.

In the remaining part of the paper we shall prove the following
asymptotics:
\be\la{exas}
  \Vert f(t)\Vert_{L^{\infty}_{-\beta}}\sim t^{-1}, \quad
  z(t)\sim t^{-1/2},\quad \Vert\rw(t)\Vert_{H^1}\sim t^{-1/2}, \quad t\to\infty.
\ee
\br\la{order}
  To justufy these asymptotics, we will separate leading terms and remainders
  in right hand side of equations (\ref{omega})-(\ref{z}), (\ref{h}).
  Namely, we shall expand  the expressions  for
  $\dot \om$,  $\dot\gamma$ and  $\dot z$ up to and including
terms of the order ${\cal O}(t^{-3/2})$,
  and for $\dot h$ up to  ${\cal O}(t^{-1})$ keeping in mind the asymptotics (\re{exas}).
  This choice is necessary for applicaton of the method of majorants.
\er

\subsection{Equation for $\om$}
\label{lt-2}
Using the equality ${\bf Q}[\chi]=jE[\chi]$, and the fact that
$j({\bf P}^0)^*={\bf P}^0j$ (where $^*$ means adjoint with respect to
the Hermitian inner product $\langle\,\cdot\, ,\,\cdot\,\rangle$),
we rewrite
$$
  \langle {\bf P}^0{\bf Q}[\chi],\Phi\rangle =\langle {\bf P}^0jE[\chi],\Phi\rangle
  =-\langle E[\chi],j({\bf P}^0)^*\Phi\rangle =-\langle E[\chi],{\bf P}^0j\Phi\rangle
$$
with $\chi=\rw+f$ and $\Phi=\Psi+\chi=\Psi+\rw+f$.
Then equation (\ref{omega}) for $\dot\om$ can be  expanded up to ${\cal O}(t^{-3/2})$,
assuming (\ref{exas}), as follows:
\beqn\la{omega1}
  \dot\omega\!\!\!&=&\!\!\!-\frac 1{\Delta}
  \Biggl[\langle E_2[\rw,\rw]+2E_2[\rw,f]+E_3[\rw,\rw,\rw],j\Psi\rangle
  +\langle E_2[\rw,\rw], {\bf P}^0j\rw\rangle\Biggr]\\
  \nonumber
  \!\!\!&+&\!\!\!\frac 1{\Delta^2}\langle E_2[\rw,\rw],j\Psi\rangle
  \Bigl(\langle\pa_\om \Psi,\rw\rangle-\langle\pa_\om {\bf P}^0\rw, \Psi\rangle\Bigr)
  +\Omega_R,
\eeqn
where
\be\la{Om}
  |\Omega_R| ={\cal R}(\om,|z|+\Vert f\Vert_{L^{\infty}_{-\beta}})(|z|^2
  +\Vert f\Vert_{L^{\infty}_{-\beta}})^{2}.
\ee
Substituting $\rw=zu+\ov zu^*$, we can write (\ref{omega1})
in the form
\beqn\la{omega2}
  \!\!\!\!\!\!\!\!\!\!\dot\om\!=\!\Omega_{20}z^2\!+\Omega_{11}z\ov z+\Omega_{02}\ov z^2\!
  +\Omega_{30}z^3\!+\Omega_{21}z^2\ov z+\Omega_{12}z\ov z^2\!
  +\Omega_{03}\ov z^3\!+\!z\langle f,\Omega'_{10}\rangle\!
  +\!\overline z\langle f,\Omega'_{01}\rangle\!+\!\Omega_R.
\eeqn

Let us now display explicitly some important terms of this expansion.
First we compute the quadratic terms in (\ref{omega1})
which are of order $t^{-1}$
according to (\ref{exas}): these  are obtained from
\be\la{imp}
  \langle E_2[\rw,\rw],j\Psi\rangle=z^2\langle E_2[u,u],j\Psi\rangle
  +\ov z^2\langle E_2[u^*,u^*],j\Psi\rangle
  +2z\ov z\langle E_2[u^*,u],j\Psi\rangle
\ee
Taking into account the definition of $E_2$, the identity $(\Psi,j\Psi)=0$,
the fact that $\Phi=(\phi,0)$, and $\rw=zu+\ov zu^*$, we obtain
\beqn\nonumber
  \langle E_2[\rw,\rw],j\Psi\rangle
  \!\!\!&=&\!\!\!\langle \delta(x)2a'(C^2)(\Psi,\rw)\rw,j\Psi\rangle
  = 2a'(C^2)(z+\ov z)(u(0),\Psi(0))(z-\ov z)(u(0),j\Psi(0))\\
  \nonumber
  \!\!\!&=&\!\!\!2(z^2-\overline z^2) a'(C^2)(u(0),\Psi(0))(u(0),j\Psi(0)),
\eeqn
where $(u,v)=u_1v_1+u_2v_2$ for $u,v\in\C^2$.
Therefore
\be\la{Omega20}
  \Omega_{20}=\overline\Omega_{02}=-\frac{\langle E_2[u,u],j\Psi\rangle} {\Delta}
  =-\frac{2}{\Delta} a'(C^2)(u(0),\Psi(0))(u(0),j\Psi(0))
\ee
is purely imaginary and
\be\la{Omega11}
  \Omega_{11}=-2\frac{\langle E_2[u,u^*],j\Psi\rangle}{\Delta}=0.
\ee
Using the property (\ref{XYZ}), we find that
\be\la{Omega10}
  \Omega'_{10}=\overline\Omega'_{01}=-2\frac{E_2[u^*,j\Psi]}{\Delta}.
\ee
\br\la{Om-sub}
  Since  $f\in X^c_t$ then
  $\langle f, \Omega'_{10}\rangle=\langle {\bf P}^cf, \Omega'_{10}\rangle
    =\langle f, j{\bf P}^cj^{-1}\Omega'_{10}\rangle$.
  Therefore we can substitute $\Omega'_{10}$
  in (\ref{omega2}) by their projection $j{\bf P}^cj^{-1}\Omega'_{10}$.
\er
\subsection{Equation for $\gamma$}
\label{lt-1}
Using again the equality ${\bf Q}=jE$ we get
$$
  \langle {\bf Q}[\chi],j(\pa_\om \Psi-\pa_\om {\bf P}^0\chi)\rangle
  =\langle E[\chi],\pa_\om \Psi-\pa_\om {\bf P}^0\chi\rangle
$$
Therefore (\ref{gamma}), (\ref{den1}), (\ref{den2}), (\ref{E}), (\ref{ER}) imply
\beqn\la{gamma1}
  \!\!\!\!\dot\gamma&=&\!\!\!\!\Delta^{-1}
  \Biggl[\langle E_2[\rw,\rw]+2E_2[\rw,f]+E_3[\rw,\rw,\rw],\pa_\om \Psi\rangle
  -\langle E_2[\rw,\rw],\pa_\om {\bf P}^0\rw\rangle\Biggr]\\
  \nonumber
  &-&\Delta^{-2}\langle E_2[\rw,\rw], \pa_\om \Psi\rangle
  \Bigl(\langle \pa_\om \Psi,\rw\rangle-\langle\pa_\om {\bf P}^0\rw, \Psi\rangle\Bigr)
  +\Gamma_R,
\eeqn
where
\be\la{Gam}
  |\Gamma_R|={\cal R}(\om,|z|+\Vert f\Vert_{L^{\infty}_{-\beta}})(|z|^2
  +\Vert f\Vert_{L^{\infty}_{-\beta}})^2
\ee
Equation (\ref{gamma1}) can thus be represented in the form
\be\la{gamma2}
  \dot\gamma=\Gamma_{20}z^2+\Gamma_{11}z\overline z+\Gamma_{02}\overline z^2
  +\Gamma_{30}z^3+\Gamma_{21}z^2\overline z+\Gamma_{12}z\overline z^2
  +\Gamma_{03}\overline z^3 +z\langle f,\Gamma'_{10}\rangle
  +\overline z\langle f,\Gamma'_{01}\rangle+\Gamma_R
\ee
where
\beqn\la{Gij}
  \!\!\!\!\!\!\!\!\!\Gamma_{20}
  \!\!&=&\!\!\frac{\langle E_2[u,u], \pa_\om \Psi\rangle}{\Delta},
  \quad\Gamma_{11}=2\frac{\langle E_2[u,u^*], \pa_\om \Psi\rangle}{\Delta},\quad
  \Gamma_{02}=\frac{\langle E_2[u^*,u^*], \pa_\om \Psi\rangle}{\Delta},\\
  \nonumber
  \!\!\!\!\!\!\!\!\!\Gamma'_{10}\!\!&=&\!\!2\frac{E_2[u^*, \pa_\om \Psi]}{\Delta},
  \quad\Gamma'_{01}=2\frac{E_2[u, \pa_\om \Psi]}{\Delta}.
\eeqn
\subsection{Equation for $z$}
\label{lt-3}
Denote $\varkappa=\langle u,ju\rangle$ and rewrite (\ref{z}) in the form:
\beqn\la{z1}
  \dot z-i\mu z&=&\frac{\langle E_2[\rw,\rw]+2E_2[\rw,f]+E_3[\rw,\rw,\rw],u\rangle}
  {\varkappa}\\
  \nonumber
  &+&\frac{\langle\pa_\om \rw,ju\rangle\langle E_2[\rw,\rw],j\Psi\rangle
  -\langle \rw,u\rangle\langle E_2[\rw,\rw], \pa_\om\Psi\rangle}{\varkappa\Delta}+Z_R,
\eeqn
where
\be\la{ZR}
  |Z_R|={\cal R}(\om,|z|+\Vert f\Vert_{L^{\infty}_{-\beta}})
  (|z|^2+\Vert f\Vert_{L^{\infty}_{-\beta}})^2.
\ee
Equation (\ref{z1}) can be represented in the form
\beqn\la{z2}
  \dot z&=&i\mu z+Z_{20}z^2+Z_{11}z\overline z+Z_{02}\overline z^2
  +Z_{30}z^3+Z_{21}z^2\overline z+Z_{12}z\overline z^2+Z_{03}\overline z^3\\
  \nonumber
  &+&z\langle f,Z'_{10}\rangle+\overline z\langle f,Z'_{01}\rangle+Z_R,
  \eeqn
where, using the calculations in the previous two sections,
we have in particular,
\beqn\nonumber
  Z_{20}&=&\frac{\langle E_2[u,u],u\rangle}{\varkappa},
  \quad Z_{11}=2\frac{\langle E_2[u,u^*],u\rangle}{\varkappa},
  \quad Z_{02}=\frac{\langle E_2[u^*\!,u^*],u\rangle}{\varkappa},\\
  \la{Zij}
  Z_{21}&=&\frac{\langle 3E_3[u^*\!,u,u],u\rangle}{\varkappa}\\
  \nonumber
  &+&\frac{\langle\pa_\om u^*\!,\!ju\rangle\langle E_2[u,u],j\Psi\rangle
  \!-\!\langle u^*\!,u\rangle\langle E_2[u,u], \pa_\om\Psi\rangle
  \!-\!\langle u,u\rangle\langle 2E_2[u^*,u], \pa_\om\Psi\rangle}
  {\varkappa\Delta},\\
  \nonumber
  Z'_{10}&=&2\frac{E_2[u^*,u]}{\ov{\varkappa}},
  \quad Z'_{01}=2\frac{E_2[u,u]}{\ov{\varkappa}}.
\eeqn

\subsection{Equation for $h$}
\label{lt-4}
We now turn to equation (\ref{h}) for $\dot h$ that we rewrite in the form
\be\la{h1}
  \dot h={\bf C}_{T}h+{\bf P}_{T}^c\Big[({\bf C}-{\bf C}_{T})f
  +{\bf P}^cjE_2[w,w]+\dot\gamma{\bf P}^cj^{-1}f+H_R\Big],
\ee
where the remainder $H_R$ is
\beqn\la{HR}
  H_R&=&{\bf P}^cj(E[\chi]-E_2[\rw,\rw])+\dot\om\pa_\om{\bf P}^c\chi
  +\dot\gamma{\bf P}^cj^{-1}\rw\\
  \nonumber
  &=&{\bf P}^cj(E[\chi]-E_2[\rw,\rw])-\dot\om\pa_\om{\bf P}^d\chi
  +\dot\gamma j^{-1}\rw -\dot\gamma {\bf P}^dj^{-1}\rw.
\eeqn
For  the  $H_R$ we have, recalling \eqref{cM}, the following estimate
\beqn\nonumber
  \Vert H_R\Vert_{{\cal M}_{\beta}}\!\!\!&=&\!\!\!
  {\cal R}(\om,|z|+\Vert f\Vert_{L^{\infty}_{-\beta}})
  (|z|^3+|z|\Vert f\Vert_{L^{\infty}_{-\beta}}+\Vert f\Vert_{L^{\infty}_{-\beta}}^2)
  +{\cal R}(\om)|\dot\om|(|z|+\Vert f\Vert_{L^{\infty}_{-\beta}})\\
  \la{HR-est}
  \!\!\!&+&\!\!\!{\cal R}(\om)|\dot\gamma||z|
  ={\cal R}(\om,|z|+\Vert f\Vert_{L^{\infty}_{-\beta}})
  (|z|^3+|z|\Vert f\Vert_{L^{\infty}_{-\beta}}+\Vert f\Vert_{L^{\infty}_{-\beta}}^2).
\eeqn
Now we continue the isolation of the leading terms in the right hand side of (\ref{h1}).
Note that
$$
  {\bf C}-{\bf C}_{T}=j^{-1}(\omega-\omega_{T})+j^{-1}(V-V_{T}),\quad {\rm where}\quad
  V=-\delta(x)[a+bP_1].
$$
Also
${\bf P}_{T}^c{\bf P}^c={\bf P}_{T}^c[{\bf P}_{T}^c+{\bf P}_{T}^d-{\bf P}^d]
  ={\bf P}_{T}^c+{\bf P}_{T}^c[{\bf P}_{T}^d-{\bf P}^d]$.
Therefore, (\ref{h1}) becomes
\be\la{h2}
  \dot h={\bf C}_{T}h+\sigma(t){\bf P}_{T}^cj^{-1}h+{\bf P}_{T}^cjE_2[w,w]+H_R'
\ee
with
$\sigma(t)=\omega-\omega_{T}+\dot\gamma$, and
$$
  H_R'={\bf P}_{T}^c[H_R+\sigma(t)j^{-1}g+j^{-1}(V-V_{T})f+({\bf P}_{T}^d-{\bf P}^d)
  j(E_2[w,w]+\dot\gamma f)].
$$
Using  the identity
${\bf P}_T^c=1- {\bf P}_T^d$, we obtain
\be\la{HR'-est}
   \Vert H_R'\Vert_{{\cal M}_{\beta}}={\cal R}(\om,|z|
   +\Vert f\Vert_{L^{\infty}_{-\beta}})
   (|z|^3+|z|\Vert f\Vert_{L^{\infty}_{-\beta}}+\Vert f\Vert_{L^{\infty}_{-\beta}}^2
   +|\om-\om_T|(|z|^2+|\dot\gamma|\Vert f\Vert_{L^{\infty}_{-\beta}}).
  \ee
Next we need an additional construction to combine first two terms in RHS
of (\ref{h2}). Namely, lemma \ref{rem2} below shows that the ``main part'' of the
second term is $i\sigma(t)({\bf\Pi}_{T}^{+}-{\bf\Pi}_{T}^{-})h$,
where
${\bf\Pi}^{+}$ and ${\bf\Pi}^{-}$ are defined in Proposition \ref{TD}.
Hence, we denote
\be\la{CM}
  {\bf C}_{M}(t)={\bf C}_{T}+i\sigma(t)({\bf\Pi}_{T}^{+}-{\bf\Pi}_{T}^{-})
\ee
and rewrite  (\ref{h2}) as
\be\la{h3}
  \dot h={\bf C}_{M}(t)h+{\bf P}_{T}^cjE_2[w,w]+\tilde H_R,
\ee
where
\be\la{HR1}
  \tilde H_R=H_R'+\sigma(t)[{\bf P}_{T}^cj^{-1}-i({\bf\Pi}_{T}^{+}-{\bf\Pi}_{T}^{-})]h
\ee

\begin{lemma}\label{rem2}
  For $h\in X_{T}^c$ we have
  \be\la{Proj}
    \Vert [{\bf P}^c_{T}j^{-1}-i({\bf\Pi}^{+}_T-{\bf\Pi}^{-}_T)]h\Vert_{L^1_{\beta}}\le
    \Vert h\Vert_{L^{\infty}_{-\beta}}.
  \ee
\end{lemma}
This lemma is proved in appendix D.
Lemma \ref{rem2} and the bound (\ref{HR'-est}) imply
\begin{pro}\label{rem1}
  The remainder $\tilde H_R$ admits the bound
  \be\la{tHR-est}
    \Vert \tilde H_R\Vert_{{\cal M}_{\beta}}={\cal R}_1(\om,|z|
    +\Vert f\Vert_{L^{\infty}_{-\beta}})
    (|z|^3+|z|\Vert f\Vert_{L^{\infty}_{-\beta}}+\Vert f\Vert_{L^{\infty}_{-\beta}}^2
    +|\om-\om_T|(|z|^2+\Vert f\Vert_{L^{\infty}_{-\beta}})).
  \ee
\end{pro}
\setcounter{equation}{0}
\section{Canonical form}
\label{eqns-trans}
Our goal is to transform the evolution equations for $(\om,\gamma, z,h)$
to a more simple, canonical form. We will use the idea of normal coordinates,
trying to keep unchanged the estimates for the remainders as much as is
possible. 
\subsection{Canonical form of  equation for $h$}
\label{f-trans}
We expand out the middle term on the right hand side
of (\ref{h3}), obtaining
\be\la{h4}
  \dot h={\bf C}_M(t) h+H_{20}z^2+H_{11}z\overline z+H_{02}\overline z^2+\tilde H_R.
\ee
Here, the coefficients $H_{ij}$ are defined by
\be\la{Hij}
  H_{20}={\bf P}^c_T jE_2[u,u],\quad H_{11}=2{\bf P}^c_T jE_2[u,u^*],
  \quad H_{12}={\bf P}^c_T jE_2[u^*,u^*].
\ee
We want to extract from $h$ the term of order $z^2\sim t^{-1}$.
For this purpose we expand $h$ as
\be\la{h-dec}
  h=h_1+k+k_1,
\ee
where
\be\la{k}
  k=a_{20}z^2+a_{11}z\overline z+a_{02}\overline z^2,
\ee
with some $a_{ij}\equiv a_{ij}(\om,x)$ satisfying $a_{ij}=\overline a_{ij}$, and
\be\la{k1}
  k_1=-\exp\Biggl(\int_0^t {\bf C}_M(\tau)d\tau\Biggr)k(0).
\ee
Note that $k_1$ is just the solution of the corresponding
homogeneous equation $\dot k_1={\bf C}_M k_1$, since the operators
${\bf C}_M(t)$ all commute for different values of $t$.
It follows from $k_1(0)=-k(0)$ that $h_1(0)=h(0)$.
\begin{lemma}\label{h-trans}
  There exist $a_{ij}\in L^{\infty}_{-\beta}$
in (\ref{k}) such that the equation
  for $h_1$ has the form
  \be\la{h5}
    \dot h_1={\bf C}_M(t) h_1+\hat H_R
  \ee
  where $\hat H_R=\tilde H_R+H'$, with
estimates as in \eqref{tHR-est}, and also
\be\la{kb}
\|k\|_{L^{\infty}_{-\beta}}={\cal R}_1(\om)|z|^2.
\ee
\end{lemma}
\bp (cf. Section 4.2.2 in \cite {BS})
  We substitute (\ref{k}) into (\ref{h3}) and equate the coefficients of the quadratic
  powers of $z$. In addition we replace  the discrete eigenvalue $\mu(t)$ by its value
  at time $T$, i.e. $\mu_T=\mu(\om(T))$, and include the correction in the remainder.
  Then we get
  $$
     H_{20}-2i\mu_Ta_{20}=-{\bf C}_Ta_{20},\quad\quad H_{11}=-{\bf C}_T a_{11},\quad\quad
     H_{02}+2i\mu_Ta_{02}=-{\bf C}_Ta_{02}
  $$
  and $\hat H_R=\tilde H_R+H'$, where $H'$ is defined as
   \beqn \la{H'}
     H'\!\!\!&=&\!\!\!\sum\pa_\om a_{ij}{\cal R}(\om,|z|
     +\Vert f\Vert_{L^{\infty}_{-\beta}})
     |z|^2 (|z|+\Vert f\Vert_{L^{\infty}_{-\beta}})^2\\
     \nonumber
     \!\!\!&+&\!\!\!\sum a_{ij}{\cal R}(\om,|z|+\Vert f\Vert_{L^{\infty}_{-\beta}})
    |z| (|z|+\Vert f\Vert_{L^{\infty}_{-\beta}})^2
    +\sum a_{ij}{\cal R}(\om)|z|^2|\mu_T-\mu|
    -i\sigma({\bf\Pi}_T^+-{\bf\Pi}_T^-)k.
  \eeqn
  The dependency in $x$ appears here through the coefficients $a_{ij}=a_{ij}(\om,x)$.
  Notice, from \eqref{e2d} and the formulae for the
projection operators in \S\ref{subspace},
that each $H_{ij}\in X^c_T$ is the sum of a multiple of
$\delta(x)$ and a function exponentially
  decreasing  at infinity.
  Hence, there exists a solution  $a_{11}$ in the form
  \be\la{a11}
    a_{11}=-{\bf C}_T^{-1}H_{11},
  \ee
  where ${\bf C}_T^{-1}$ stands for regular part of the resolvent $R(\lam)$
  at $\lam=0$ since the singular part of $R(\lam)H_{11}$ vanishes for $H_{11}\in X^c_T$.
  The function $a_{11}$ is exponentially decreasing at infinity.\\
  For $a_{20}$ and $a_{02}$ we  choose  the following inverse operators:
  \be\la{a20}
     a_{20}=-({\bf C}_T-2i\mu_T-0)^{-1}H_{20},\quad
     a_{02}={\overline a}_{20}=-({\bf C}_T+2i\mu_T-0)^{-1}H_{02},
  \ee
  This choice is motivated by Proposition \ref{l1}, and putting
$t=0$ in that Proposition we have the bound \eqref{kb}.
The remainder $H'$ can be written as
  \be\la{Am}
     H'=\sum\limits_m({\bf C}_T-2i\mu_Tm-0)^{-1}A_m,\quad m\in\{-1,0,1\}
  \ee
  with   $A_m\in X^c_T$, satisfying the estimate
  \be\la{Am-est1}
    \Vert A_{m}\Vert_{L^1_\beta}={\cal R}(\om,|z|+\Vert f\Vert_{L^{\infty}_{-\beta}})
    |z|\Bigl(|z||\om_T-\om|+(|z|+\Vert f\Vert_{L^{\infty}_{-\beta}})^2\Bigr).
  \ee
\ep
\subsection{Canonical form of equation for $\om$}
\label{om-trans}
We want to remove all terms in the right hand side of (\ref{omega2}) except the remainder
$\Omega_R$. This is possible by methods of \cite{BS}
since $\Omega_{11}=0$ by (\ref{Omega11}).

\begin{lemma}\label{bij-exist}
  There exist coefficients $b_{ij}(\om)$, $0\le i,j\le 3$, and  
exponentially decreasing functions
$b'_{ij}(x,\om)$, $0\le i,j\le 1$, such that real-valued
function $\om_1$ defined as
  \be\la{new-om}
    \om_1=\om+b_{20}z^2+b_{02}\ov z^2+b_{30}z^3+b_{21}z^2\ov z+b_{12}z\ov z^2
    +b_{03}\ov z^3+z\langle f,b'_{10}\rangle+\overline z\langle f,b'_{01}\rangle\\
  \ee
  obeys a differential equation of the form
  \be\la{n-om}
    \dot\om_1=\hat\Omega_R,
  \ee
  where $\hat\Omega_R$ satisfies the same estimate (\ref{Om}) as $\Omega_R$:
  \be\la{hOm}
    |\hat\Omega_R|={\cal R}(\om,|z|+\Vert f\Vert_{L^{\infty}_{-\beta}})
    (|z|^2+\Vert f\Vert_{L^{\infty}_{-\beta}})^2
  \ee
\end{lemma}
\bp
  The calculation follows the classical method of normal coordinates.
  Substituting $\dot\om$, $\dot z$,  and $\dot f$ from
  (\ref{omega2}), (\ref{z2}), (\ref{f}) into the equation for $\dot\om_1$ and
  comparing  the coefficients of $z^2$, $zf$, etc.
  leads to a system of equations for the coefficients $b_{20}$,  $b'_{10}$, ets.
  (cf. \cite[Proposition 4.1]{BS})
  \beqn \nonumber
    &&\Omega_{20}+2i\mu b_{20}=0,\\
    \nonumber
    &&\Omega'_{10}+i\mu b'_{10}+{\bf C}^*b'_{10}=0,\\
   \la{system}
    &&\Omega_{21}+2Z_{11}b_{20}+i\mu b_{21}+2Z_{20}b_{02}+
    \langle F_{11},b'_{10}\rangle+\langle F_{20},b'_{01}\rangle=0,\\
    \nonumber
    &&\Omega_{30}+2Z_{20}b_{20}+3i\mu b_{30}+\langle F_{20},b'_{10}\rangle=0.
    \eeqn
  From the first equation of (\ref{system}) we obtain
  \be\la{b20}
      b_{20}=\ov b_{02}=\frac{i}{2\mu}\Omega_{20},
  \ee
  Multiply the second equation of (\ref{system}) by $j$ we get
  $j\Omega'_{10}+i\mu jb'_{10}-{\bf C}jb'_{10}=0$
  since $j{\bf C}^*=-{\bf C}j$.
  Without loss of generality we can assume that $j\Omega'_{10}\in X^c_t$
  by Remark \ref{Om-sub}. Therefore, there exists the solution $ b'_{10}$ in the form
  \be\la{b'10}
      b'_{10}=\ov{b'}_{01}=-j({\bf C}-i\mu)^{-1}j\Omega'_{10},
  \ee
  where $({\bf C}-i\mu)^{-1}$ stands for regular part of the resolvent $R(\lam)$
  at $\lam=i\mu$ since the singular part of $R(\lam)j\Omega'_{10}$ vanishes
  for $j\Omega'_{10}\in X^c_T$.
  The functions $b'_{10}$, $b'_{01}$ decrease  exponentially
at infinity, and
the equations for $b_{21}=\ov b_{12}$,  $b_{30}=\ov b_{03}$ can
be easily solved.
\ep
\subsection{Canonical form of  equation for $z$}
\label{z-transform}
In this section we obtain a canonical form of the equation (\ref{z2}) for $z$,
and carry out a computation of the coefficient of the resonant ``$z^2\overline{z}$ '' term,
which gives the Fermi Golden Rule. Substituting
(\ref{gh}) and (\ref{h-dec}) into (\ref{z2}) and putting the contribution of
$g+h_1+k_1$ in the remainder $\tilde Z_R$, we obtain
\beqn\la{z3}
  \dot z&=&i\mu z+Z_{20}z^2+Z_{11}z\overline z+Z_{02}\overline z^2
  +Z_{30}z^3+Z_{21}z^2\overline z+Z_{12}z\overline z^2+Z_{03}\overline z^3\\
  \nonumber
  &+&Z_{30}'z^3+Z_{21}'z^2\overline z+Z_{12}'z{\overline z}^2
  +Z_{03}'{\overline z}^3+\tilde Z_R.
\eeqn
We have by (\ref{h-dec})-(\ref{k})
\be\la{Z3}
  Z_{30}'=\!\langle a_{20},Z'_{10}\rangle,~~
  Z_{21}'=\!\langle a_{11},Z'_{10}\rangle\!+\!\langle a_{20},Z'_{01}\rangle,~~
  Z_{03}'=\!\langle a_{02},Z'_{01}\rangle,~~
  Z_{12}'=\!\langle a_{02},Z'_{10}\rangle\!+\!\langle a_{11},Z'_{01}\rangle
\ee
We are particularly interested in the resonant term  $Z'_{21}z^2\ov z$.
Formulas (\ref{Zij}), (\ref{Hij}), (\ref{a11}), (\ref{a20}) imply
\be\la{Z21}
  Z'_{21}=-\langle {\bf C}_T^{-1}2{\bf P}_T^cj E_2[u,u^*],
  2\frac{E_2[u,u^*]}{\ov\varkappa}\rangle
  -\langle ({\bf C}_T-2i\mu_T-0)^{-1}{\bf P}_T^cjE_2[u,u],
  2\frac{E_2[u,u]}{\ov\varkappa}\rangle.
\ee
For the coefficient $\varkappa=\varkappa(\om)$ we get (see \cite[Proposition 3.1]{BS})
\be\la{uu}
    \varkappa=\langle u,ju\rangle=i\delta,\quad {\rm with}\quad \delta>0.
\ee
\begin{lemma}\label{pv}
  Suppose that the non-degeneracy condition (\ref{FGR})  is satisfied, then
  \be\la{ReZ}
    \rRe Z'_{21}<0
  \ee
  for $\om$ in some vicinity of $\om_0$.
\end{lemma}
\bp
  We first notice that the coefficient
  $\langle {\bf C}_T^{-1}2{\bf P}_T^cj E_2[u,u^*],E_2[u,u^*]\rangle$
  appearing in the expression (\ref{Z21}) for $Z'_{21}$ is real,
  since the operator ${\bf C}_T^{-1}2{\bf P}_T^cj$ is selfadjoint.
  Hence by (\ref{uu}) $\rRe Z'_{21}$ reduces to
 $$
    \rRe Z'_{21}=-\rRe 2\frac{\langle ({\bf C}_T\!-2i\mu_T-0)^{-1}{\bf P}_T^c
    jE_2[u,u], E_2[u,u]\rangle}{\varkappa}
    =-\frac 2{\delta}\rIm\langle R(2i\mu+0){\bf P}_T^c
    jE_2[u,u], E_2[u,u]\rangle
 $$
  where we denote
  $R(\lam)=R_T(\lam)=({\bf C}_T-\lam)^{-1}$, $\rRe\lam>0$ and $\mu=\mu_T$.
  Using that ${\bf P}_T^c$ commutes with $ R(2i\mu+0)$, we have
  $ R(2i\mu+0){\bf P}_T^c={\bf P}_T^c R(2i\mu+0){\bf P}_T^c$.
  We have also that $({\bf P}_T^c)^*=-j{\bf P}_T^cj$, hence
  $$
    \rRe Z'_{21}=\frac 2{\delta}\rIm\langle  R(2i\mu+0)\al, j\al\rangle
  ~~~{\rm with}~~~~\al={\bf P}_T^cjE_2[u,u].
  $$
 The function $\lam\mapsto\langle R(\lam)\al,j\al\rangle$ is analytic
  in the region $\C\setminus ({\cal C}_+\cup{\cal C}_-)$ since  $\al\in X_T^c$.
  Hence by the Cauchy residue theorem we have
  \be\la{sign1}
     \langle  R(2i\mu+0)\al, j\al\rangle=
     -\frac 1{2\pi i}\int\limits_{{\cal C}_+\cup{\cal C}_-}d\lam~
     \frac {\langle (R(\lam+0)-R(\lam-0))\al, j\al\rangle}{\lam-2i\mu-0}
  \ee
  Now we use the representation
  \be\la{R-rep}
    R(\lam+0)-R(\lam-0)=-\frac {\tau_{\pm}(\lam)\otimes \ov \tau_{\pm}(\lam)}
    {8ik_{\pm}D\ov D}j-\frac {s_{\pm}(\lam)\otimes \ov s_{\pm}(\lam)}
    {2ik_{\pm}}j,\quad \lam\in {\cal C}_{\pm} 
  \ee
  where $D=D(\lam+0)$, $k_{\pm}=k_{\pm}(\lam+0)$, and
  $\tau_{\pm}(\lam)$ and $s_{\pm}(\lam)$
  are the even and the odd eigenfunctions of the operator ${\bf C}_T$
  corresponding to  $\lam\in {\cal C}_{\pm}$ (see appendix B).
  The representation (\ref{R-rep})  can be checked by direct calculation
  using formulas (\ref{dec})-(\ref{tau1}) for the resolvent.
  Then equation (\ref{sign1}) becomes
  \be\nonumber
     \langle R(2i\mu+0)\al, j\al\rangle
     =-\frac 1{16\pi }\int_{-\infty}^{-\om}\;\frac{d\nu}{k_-|D|^2}
     \frac{\langle\tau_-,j\al\rangle\ov{\langle\tau_-,j\al\rangle}}{\nu-2\mu}
     -\frac 1{16\pi }\int_{\om}^{\infty}\;\frac{d\nu}{k_+|D|^2}
     \frac{\langle\tau_+,j\al\rangle\ov{\langle\tau_+,j\al\rangle}}{\nu-2\mu+i0}
  \ee
  since the function $\al$ is even. Using that
  $
    \frac 1{\nu+i0}=p.v.\frac 1{\nu}-i\pi\delta(\nu)
  $
  where p.v. is the Cauchy principal value, we have
  \beqn\nonumber
   \langle  R(2i\mu+0)\al, j\al\rangle
     \!\!\!&=&\!\!\!-\frac 1{16\pi }\int_{-\infty}^{-\om}\;\frac{d\nu}{k_-|D|^2}
     \frac{\langle\tau_-,j\al\rangle\ov{\langle\tau_-,j\al\rangle}}{\nu-2\mu}
     -\frac 1{16\pi }p.v.\int_{\om}^{\infty}\;\frac{d\nu}{k_+|D|^2}
     \frac{\langle\tau_+,j\al\rangle\ov{\langle\tau_+,j\al\rangle}}{\nu-2\mu}\\
     \nonumber
     \!\!\!&+&\!\!\!\frac i{16}\frac{\langle\tau_+(2i\mu),j\al\rangle
     \ov{\langle\tau_+(2i\mu),j\al\rangle}}{k_+(2i\mu+0)|D(2i\mu+0)|^2}.
  \eeqn
  Since the integral terms are real, this implies
  $$
    \rIm\langle  R_T(2i\mu_T+0)\al, j\al\rangle=
    \frac{|\langle\tau_+(2i\mu_T), E_2[u,u]\rangle|^2}
    {16k_+(2i\mu_T+0)|D(2i\mu_T+0)|^2}.
  $$
  The non-degeneracy condition (\ref{FGR}) implies that
  $\langle\tau_+(2i\mu_T), E_2[u,u]\rangle\not =0$ in some vicinity of
  $\om_0$. Using also the inequality $k_+(2i\mu_T+0)<0$,
we deduce $\rRe Z'_{21}<0$.
\ep
Now we estimate  the remainder $\tilde Z_R$.
\begin{lemma}\label{tiZR}
  The remainder $\tilde Z_R$ has the form
  \be\la{tZR}
    \tilde Z_R={\cal R}_1(\om,|z|+\Vert f\Vert_{L^{\infty}_{-\beta}})
    \Bigl[(|z|^2+\Vert f\Vert_{L^{\infty}_{-\beta}})^2+
    |z||\om_T-\om|\Vert h\Vert_{L^{\infty}_{-\beta}}
    +|z|\Vert k_1\Vert_{L^{\infty}_{-\beta}}+|z|\Vert h_1\Vert_{L^{\infty}_{-\beta}}\Bigr]
  \ee
\end{lemma}
\bp
  The remainder $\tilde Z_R$ is given by
  $$
    \tilde Z_R=Z_R+z\langle f-k,Z'_{10}\rangle+\overline z\langle f-k,Z'_{01}\rangle
  $$
   where $Z_R$ satisfies estimate (\ref{ZR}).
  Since $f-k=g+k_1+h_1$, we have by (\ref{g})
  $$
    |\langle f-k,Z'_{10}\rangle|\le
    {\cal R}(\om)(\Vert g\Vert_{L^{\infty}_{-\beta}}
    +\Vert k_1\Vert_{L^{\infty}_{-\beta}}+\Vert h_1\Vert_{L^{\infty}_{-\beta}})
    \le{\cal R}_1(\om)(|\om_T-\om|\Vert h\Vert_{L^{\infty}_{-\beta}}
    +\Vert k_1\Vert_{L^{\infty}_{-\beta}}+\Vert h_1\Vert_{L^{\infty}_{-\beta}})
  $$
  which implies (\ref{tZR}).
  \ep
We can apply now the method of normal coordinates to equation (\ref{z3}).
\begin{lemma}\label{nc}  (cf. \cite[Proposition 4.9]{BS})\\
  There exist coefficients $c_{ij}$ such that the new function $z_1$ defined by
  \be\la{new-z}
    z_1=z+c_{20}z^2+c_{11}z\ov z+c_{02}\ov z^2+c_{30}z^3
    +c_{12}z\ov z^2+c_{03}\ov z^3,
  \ee
satisfies an equation of the form
 $\dot z_1=i\mu(\om)z_1+iK(\om)|z_1|^2z_1+\hat Z_R$,
 where $\hat Z_R$ satisfies estimates of the same type as $\tilde Z_R$, and
  \be\la{Re K}
    \rRe(iK)=\rRe Z'_{21}<0.
  \ee
\end{lemma}
\bp
  Substituting $z_1$ in equation (\ref{z3}) for $z$
  and equating the coefficients, we get, in particular,
 $c_{20}=\frac{i}{\mu}Z_{20}$, $c_{11}=-\frac{i}{\mu}Z_{11}$,
 $c_{02}=-\frac{i}{3\mu}Z_{02}$ and
 $iK=Z_{21}+ Z_{21}'+c_{11}Z_{20}+2c_{20}Z_{11}+c_{11}\ov Z_{11}+2c_{02}\ov Z_{02}$.
 It is easy to check that all of the coefficients $Z_{11}$, $Z_{20}$, $Z_{02}$, and $Z_{21}$
  defined in (\ref{Zij}) are pure imaginary, and hence
  (\ref{Re K}) follows immediately.
\ep
Denoting $K_T=K(\om_T)$, the equation for $z_1$ is rewritten as
\be\la{z6}
  \dot z_1=i\mu z_1+iK_T|z_1|^2z_1+\widehat{\widehat Z}_R
\ee
where
\beqn\nonumber
 |\widehat{\widehat Z}_R|\!\!\!&\le&\!\!\!|\hat Z_R|+{\cal R}_1(\om,|z|
  +\Vert f\Vert_{L^{\infty}_{-\beta}}\!)|z|^3|\om_T-\om|
  ={\cal R}_1(\om,|z|+\Vert f\Vert_{L^{\infty}_{-\beta}}\!)
  \!\Bigl[(|z|^2+\Vert f\Vert_{L^{\infty}_{-\beta}}\!)^2\\
  \la{hhZR}
  &+&|z||\om_T-\om|(\Vert h\Vert_{L^{\infty}_{-\beta}}+|z|^2)
  +|z|\Vert k_1\Vert_{L^{\infty}_{-\beta}}
  +|z|\Vert h_1\Vert_{L^{\infty}_{-\beta}}\Bigr].
\eeqn
It is easier to deal with $y=|z_1|^2$, rather than $z_1$, because $y$ decreases
at infinity while $z_1$ is oscillating.
The equation satisfied by $y$ is simply obtained by multiplying (\ref{z6}) by $\ov z_1$
and taking the real part:
\be\la{y}
  \dot y=2\rRe(iK_T)y^2 +Y_R,
\ee
where
$$
  |Y_R|\!=\!{\cal R}_1(\om,|z|+\Vert f\Vert_{L^{\infty}_{-\beta}})|z|
  \Bigl[(|z|^2+\Vert f\Vert_{L^{\infty}_{-\beta}})^2+
  |z||\om_T-\om|(\Vert h\Vert_{L^{\infty}_{-\beta}}+|z|^2)
  +|z|\Vert k_1\Vert_{L^{\infty}_{-\beta}}
  +|z|\Vert h_1\Vert_{L^{\infty}_{-\beta}}\Bigr]
$$
The negativity of $\rRe(iK_T)$ is a key point in the analysis and was proved in
Lemma \ref{pv}.
\subsection{Canonical form of  equation for $\gamma$}
\label{gamma-transform}
The only difference between equations (\ref{omega2}) and (\ref{gamma2})
for $\om$ and $\gamma$ is that, in general the coefficient $\Gamma_{11}\not=0$.
We can nevertheless perform the same change of variables as for $\om$, obtaining:
\begin{lemma}\la{dij-exist}
There exist coefficients $d_{ij}(\om)$, $0\le i,j\le 3$,  and exponentially decreasing functions 
$d'_{ij}(x,\om)$ such that the new function $\gamma_1$ defined as
\be\la{gamma1-def}
\gamma_1=\gamma+d_{20}z^2+d_{02}\ov z^2+d_{30}z^3+d_{12}z^2\ov z+d_{12}z\ov z^2
+d_{023}\ov z^3+z\langle f,d'_{10}\rangle+\ov z\langle f,d'_{01}\rangle,
\ee
with $d_{ij}=\ov d_{ji}$, is a solution of the differential equation
\be\la{gamma3}
\dot\gamma_1=\Gamma_{11}(\om)z\ov z+\hat\Gamma_R.
\ee
Furthermore  $\hat\Gamma_R$ satisfies the same estimate (\ref{Gam}) as $\Gamma_R$.
\end{lemma}
\subsection{Bound for $|\om_T-\om|$ and  initial conditions}
\label{om-est}
Now we obtain a uniform bound for $|\om_T-\om(t)|$ on the interval $[0,T]$. 
\beqn\nonumber
  &&\!\!\!\!\!\!\!\!
  |\om_T-\om(t)|\le|\om_{1T}-\om_1(t)|+|\om_{1T}-\om_T|+|\om_1(t)-\om(t)|\\
  \nonumber
  &\!\!\!\!\!\le&\!\!\!\!\int_t^T|\dot\om_1(\tau)|d\tau+
  {\cal R}(\om_T,|z_T|+\Vert f_T\Vert_{L^{\infty}_{-\beta}})
  (|z_T|+\Vert f_T\Vert_{L^{\infty}_{-\beta}})^2
  +{\cal R}(\om,|z|+\Vert f\Vert_{L^{\infty}_{-\beta}})
  (|z|+\Vert f\Vert_{L^{\infty}_{-\beta}})^2\\
  \nonumber
  &\!\!\!\!\!\le&\!\!\!\!\max\limits_{0\le t\le T}{\cal R}
  (\om,|z|+\Vert f\Vert_{L^{\infty}_{-\beta}})
  \Bigl[\int_t^T(|z|^2+\Vert f\Vert_{L^{\infty}_{-\beta}})^2d\tau+
  (|z_T|+\Vert f_T\Vert_{L^{\infty}_{-\beta}})^2
  +(|z|+\Vert f\Vert_{L^{\infty}_{-\beta}})^2\Bigr],
\eeqn
where $z_T=z(T)$, $f_T=f(T)$.
Using
$
  |\om|\le|\om_0|+|\om_0-\om_T|+|\om-\om_T|
$, we have
$$
  \max\limits_{0\le t\le T}{\cal R}(\om,|z|+\Vert f\Vert_{L^{\infty}_{-\beta}})
  ={\cal R}(\max\limits_{0\le t\le T}|\om-\om_T|, \max\limits_{0\le t\le T}
  (|z|+\Vert f\Vert_{L^{\infty}_{-\beta}})).
$$
We denote such quantities by the symbol
${\cal R}_2(\om,|z|+\Vert f\Vert_{L^{\infty}_{-\beta}})$.
Then
\be\la{om-dif}
 |\om_T-\om|\le{\cal R}_2(\om, |z|+\Vert f\Vert_{L^{\infty}_{-\beta}})
  \Bigl[\int_t^T\!(|z|^2+\Vert f\Vert_{L^{\infty}_{-\beta}})^2d\tau+
  (|z_T|+\Vert f_T\Vert_{L^{\infty}_{-\beta}})^2
  +(|z|+\Vert f\Vert_{L^{\infty}_{-\beta}})^2\Bigr]
\ee
As in  (\ref{close}),
we suppose the smallness condition:
\be\la{zf0}
   |z(0)|\le\ve^{1/2},\quad \Vert f(0)\Vert_{L^1_\beta}\le c\ve^{3/2},
\ee
where $\ve>0$ is sufficiently small.
Equation (\ref{new-z}) implies $|z_1|^2\le |z|^2+{\cal R}(\om,|z|)|z|^3$.
Therefore
\be\la{y0}
  y(0)=|z_1(0)|^2\le\ve+{\cal R}(\om,|z(0)|)\ve^{3/2}.
\ee
From the formula $h={\bf P}_T^cf=f+({\bf P}^d-{\bf P}_T^d)f$, we see that
\be\la{h0}
  \Vert h(0)\Vert_{L^1_\beta}\le c\ve^{3/2}+{\cal R}_1(\om)|\om_T-\om|
  \Vert f(0)\Vert_{L^{\infty}_{-\beta}}.
\ee
\subsection{Bound for $k_1$}
\label{k1-est}
\begin{lemma}\la{k1-l}
  The function $k_1$ defined in (\ref{k1}) satisfies the following bound:
  \be\la{k1-1}
    \Vert k_1\Vert_{L^{\infty}_{-\beta}}\le c|z(0)|^2\frac 1{(1+t)^{3/2}}
    \le c\frac{\ve}{(1+t)^{3/2}}.
  \ee
\end{lemma}
\bp
  Equalities (\ref{CM}) and (\ref{k1}) imply
  \be\la{k1-2}
    k_1=-e^{\int_0^t{\cal C}_M(\tau)d\tau}k(0)
    =-e^{{\bf C}_Tt+i\int_0^t\beta(\tau)d\tau({\bf\Pi}_T^+-{\bf\Pi}_T^-)}k(0)
  \ee
Denoting $\nu=\int_0^t\beta(\tau)d\tau$, we obtain by expanding the
exponential,  using the idempotency of projections and
${\bf\Pi}^{+}_T+{\bf\Pi}^{-}_T+{\bf P}^d_T=1$ (the Euler trick):
  $$
   e^{i\nu{\bf\Pi}^{\pm}_T}={\bf\Pi}^{\pm}_Te^{i\nu}+{\bf\Pi}^{\mp}_T+{\bf P}^d_T.
  $$
  Therefore
  $e^{i\nu({\bf\Pi}^+_T-{\bf\Pi}^-_T)}=({\bf\Pi}^{+}_Te^{i\nu}+{\bf\Pi}^{-}_T
   +{\bf P}^d_T)({\bf\Pi}^{-}_Te^{-i\nu}+{\bf\Pi}^{+}_T+{\bf P}^d_T)
   ={\bf\Pi}^+_Te^{i\nu}+{\bf\Pi}^-_Te^{-i\nu}+{\bf P}^d_T$.\\
  Note that ${\bf C}_T$ commutes with $P^{\pm}_T$, hence
  \be\la{CM-int}
    e^{\int_0^t{\bf C}_M(\tau)d\tau}=
    e^{{\bf C}_Tt}(e^{i\nu}{\bf\Pi}^+_T+e^{-i\nu}{\bf\Pi}^-_T+{\bf P}^d_T).
  \ee
  Since $\beta$ is a real function, both exponentials are bounded.
  Further, by (\ref{k}) we have
  $k(0)=a_{20}z^2(0)+a_{11}z(0)\ov z(0)+a_{02}\ov z^2(0)$
  with $a_{ij}$ defined in (\ref{a11}), (\ref{a20}).
  Therefore, the bounds (\ref{b1}), (\ref{b5}), and assumption (\ref{zf0})
  imply (\re{k1-1}).
  \ep
\section{Large time asymptotics}
\setcounter{equation}{0}
\label{maj}
In this section we will make use of the dispersive estimates
given in \S\ref{rpsec} to prove the asymptotic representation
for the solution of \eqref{SV}
with initial data as in Theorem \ref{main}.
\subsection{Definition of majorants}
\label{maj-def}
We define the quantities
\beqn\la{M0}
   \M_0(T)&=&\max\limits_{0\le t\le T}|\om_T-\om|\Big(\frac{\ve}{1+\ve t}\Big)^{-1}\\
\la{M1}
   \M_1(T)&=&\max\limits_{0\le t\le T}|z(t)|\Big(\frac{\ve}{1+\ve t}\Big)^{-1/2}\\
\la{M2}
   \M_2(T)&=&\max\limits_{0\le t\le T}\Vert h_1\Vert_{L^{\infty}_{-\beta}}
   \Big(\frac{\ve}{1+\ve t}\Big)^{-3/2}
\eeqn
which will be refered to in the following as ``majorants'',
and denote $\M$ the 3-dimensional vector
$(\M_0,\M_1,\M_2)$.
The goal of this section is to prove that if $\ve$ is sufficiently small, $\M$ is
bounded uniformly in $T$. 
\subsection{Estimates of  remainders and initial data}
\label{maj-rem}
\begin{lemma}\la{YR-est}
  The remainder $Y_R$ defined in (\ref{y}) satisfies the estimate
  \be\la{YRest}
    |Y_R|={\cal R}(\ve^{1/2}\M)\frac{\ve^{5/2}}{(1+\ve t)^2\sqrt{\ve t}}
    (1+|\M|)^5.
  \ee
\end{lemma}
\bp
  Using  the equality $f=g+h=g+k+k_1+h_1$, Lemma \ref{k1-l} and the
definitions of the $\M_j$,
the remainder $Y_R$ is bounded as follows:
  $$
    Y_R={\cal R}_2(\om,|z|+\Vert f\Vert_{L^{\infty}_{-\beta}})|z|
    \Bigl[(|z|^2+\Vert k_1\Vert_{L^{\infty}_{-\beta}}
    +\Vert h_1\Vert_{L^{\infty}_{-\beta}})^2+
    |z||\om_T-\om|(|z|^2+\Vert k_1\Vert_{L^{\infty}_{-\beta}}
    +\Vert h_1\Vert_{L^{\infty}_{-\beta}})
  $$
  $$
    +|z|(\Vert k_1\Vert_{L^{\infty}_{-\beta}}
    +\Vert h_1\Vert_{L^{\infty}_{-\beta}})\Bigr]\!
    ={\cal R}(\ve^{1/2}\M)\Big(\frac{\ve}{1+\ve t}\Big)^{1/2}
    \M_1\Big[\!\Big(\frac{\ve}{1+\ve t}\M_1^2+\frac{\ve}{(1+ t)^{3/2}}
    +\Big(\frac{\ve}{1+\ve t}\Big)^{3/2}\!\M_3\Big)^2
  $$
  $$
    \!\!\!\!\!\!\!\!\!\!\!\!\!\!\!\!\!\!\!\!\!\!\!\!\!\!\!\!\!\!\!\!\!\!\!
    \!\!\!\!\!\!\!\!\!\!\!\!\!\!\!\!\!\!\!\!\!\!\!\!\!\!\!\!\!\!\!\!\!\!
    +\Big(\frac{\ve}{1+\ve t}\Big)^{3/2}\M_0\M_1
    \Big(\frac{\ve}{1+\ve t}\M_1^2+\frac{\ve}{(1+ t)^{3/2}}
    +\Big(\frac{\ve}{1+\ve t}\Big)^{3/2}\M_3\Big)
  $$
  $$
   +(\frac{\ve}{1+\ve t}\Big)^{1/2}
    \M_1\Big(\frac{\ve}{(1+ t)^{3/2}}
   +\Big(\frac{\ve}{1+\ve t}\Big)^{3/2}\M_3\Big)\Big]\\
   ={\cal R}(\ve^{1/2}\M)\frac{\ve^{5/2}}
    {(1+\ve t)^2\sqrt{\ve+\ve t}}(1+|\M|)^5,
  $$
establishing \eqref{YR-est}.\ep
Let us turn now to the remainder $\hat H_R=\tilde H_R+H'$ in equation
(\ref{h5}) for $h_1$.
\begin{lemma}\la{HH}
  The first summand $\tilde H_R$ satisfies
  \be\la{fHR}
    \Vert \tilde H_R\Vert_{{\cal M}_{\beta}}={\cal R}(\ve^{1/2}\M)
    \Big(\frac{\ve}{1+\ve t}\Big)^{3/2}
     \Big((1+\M_1)^3+\ve^{1/2}(1+|\M|)^4\Big).
  \ee
\end{lemma}
\bp
   It follows from (\ref{tHR-est})
   $$
     \Vert \tilde H_R\Vert_{{\cal M}_{\beta}}=
     {\cal R}_2(\om,|z|+\Vert f\Vert_{L^{\infty}_{-\beta}})
     \Bigl[|z|^3 +(|z|+|\om_T-\om|)(|z|^2+\Vert k_1\Vert_{L^{\infty}_{-\beta}}
     +\Vert h_1\Vert_{L^{\infty}_{-\beta}})
  $$
  $$
     +(|z|^2+\Vert k_1\Vert_{L^{\infty}_{-\beta}}
     +\Vert h_1\Vert_{L^{\infty}_{-\beta}})^2\Big]
     ={\cal R}(\ve^{1/2}\M)\Big[\Bigl(\frac{\ve}{1+\ve t}\Bigr)^{3/2}\M_1^3
     +\Big(\Big(\frac{\ve}{1+\ve t}\Bigr)^{1/2}\M_1+\frac{\ve}{1+\ve t}\M_0\Big)
   $$
   $$
      \Bigl(\frac{\ve}{1+\ve t}\M_1^2+\frac{\ve}{(1+t)^{3/2}}
     +\Bigl(\frac{\ve}{1+\ve t}\Bigr)^{3/2}\M_2\Bigr)
     +\Bigl(\frac{\ve}{1+\ve t}\M_1^2+\frac{\ve}{(1+t)^{3/2}}
     +\Bigl(\frac{\ve}{1+\ve t}\Bigr)^{3/2}\M_2\Bigr)^2\Big]
   $$
   which implies (\ref{fHR}).
\ep

The second summand $H'$ is represented as in (\ref{Am})
where the $A_m$ are estimated in (\ref{Am-est1}).
For the $A_m$ we now obtain:
\begin{lemma}\la{fAm}
  \be\la{Am-final}
     \Vert A_{m}\Vert_{{\cal M}_{\beta}}= {\cal R}(\ve^{1/2}\M)
     \Big(\frac{\ve}{1+\ve t}\Big)^{3/2}\Big(\M_1^3+\ve^{1/2}(1+|\M|)^3\Big).
  \ee
\end{lemma}
\bp
  Estimate (\ref{Am-est1}) implies
  $$
    \Vert A_{m}\Vert_{{\cal M}_{\beta}}
    ={\cal R}_2(\om,|z|+\Vert f\Vert_{L^{\infty}_{-\beta}})|z|
     \Big(|z||\om_T-\om|+(|z|+\Vert k_1\Vert_{L^{\infty}_{-\beta}}
     +\Vert h_1\Vert_{L^{\infty}_{-\beta}})^2\Big)
  $$
  $$
   = {\cal R}(\ve^{1/2}\M)
    \Big(\frac{\ve}{1+\ve t}\Big)^{1/2}\!\M_1\Big[\!\Big(\frac{\ve}{1+\ve t}\Big)^{3/2}
    \M_0\M_1+\Bigl(\!\Big(\frac{\ve}{1+\ve t}\Big)^{1/2}\!\M_1+\frac{\ve}{(1+t)^{3/2}}
     +\Big(\frac{\ve}{1+\ve t}\Big)^{3/2}\!\M_2\!\Big)^2\Big]
  $$
   which implies (\ref{Am-final}).
\ep

Now we  estimate the initial data. Referring to the formulas
at the end of  \S\ref{om-est}, we have
\be\la{fy0}
  y(0)\le\ve+{\cal R}(\ve^{1/2}\M)\ve^{3/2}=\ve(1+{\cal R}(\ve^{1/2}\M)\ve^{1/2}).
\ee
\be\la{fh0}
  \Vert h(0)\Vert_{{\cal M}_{\beta}}\le c\ve^{3/2}+{\cal R}_1(\om)|\om_T-\om|
  \Vert f(0)\Vert_{L^{\infty}_{-\beta}}\le c\ve^{3/2}+{\cal R}(\ve^{1/2}\M)
  \ve^2\M_0(1+\M_1^2+\ve^{1/2}\M_2).
\ee
\subsection{Integral inequalities and decay in time}
\label{can-sys}
This section is devoted to a study of the system:
\be\la{my}
  \dot y=2\rRe(iK_T)y^2+Y(t),
\ee
\be\la{mh}
  \dot h_1={\bf C}_Mh_1+H(x,t),
\ee
under some assumptions on the initial data, and on the inhomogeneous
(or source) terms $Y$ and $H$. Equation (\ref{my}) for $y$ is of Ricatti type.
For the initial data, we assume
\be\la{myh0}
  y(0)\le\ve y_0,\quad \Vert h_1(0)\Vert_{{\cal M}_{\beta}}\le \ve^{3/2}h_0
\ee
with some constant $y_0$ and $h_0>0$.
As for the source terms, we assume that
\be\la{mY}
  |Y(t)|\le {\ov Y}\frac{\ve^{5/2}}{(1+\ve t)^2\sqrt{\ve t}}
\ee
and  that $H(x,t)=H_1(x,t)+H_2(x,t)$, where
$H_2=\sum\limits_{m}({\bf C}_T-2i\mu_Tm-0)^{-1}{\cal A}_{m}$,
${\cal A}_{m}\in X^c_T$  with the following bounds:
\be\la{mH}
  \Vert H_1\Vert_{{\cal M}_{\beta}}\le\ov H_1\Bigl(\frac{\ve}{1+\ve t}\Bigr)^{3/2},
~~~~~~\Vert {\cal A}_{m}\Vert_{{\cal M}_{\beta}}\le\ov A_{m}
    \Bigl(\frac{\ve}{1+\ve t}\Bigr)^{3/2}
\ee
where  the quantities $\ov Y$, $\ov H_1$, $\ov A_m$ are supposed to be
given positive constants.
All these assumptions are motivated by the estimates of the remainders
in \S\ref{maj-rem},
and by the final estimates we intend to prove on $\om$, $z$, $h$ and $h_1$.
Equation (\ref{my}) corresponds to equation (\ref{y})
 and the assumption (\ref{mY}) on the source term has the form of estimate
(\ref{YRest}) for the remainder $Y_R$.
Similarly, equation (\ref{mh}) corresponds to equation (\ref{h5}) and assumptions
(\ref{mH}) correspond to the inequalities (\ref{fHR})- (\ref{Am-final}). 
Finally, corresponding to \eqref{Re K}, we work under the
assumption $\rRe(iK_T)=-\rIm K_T<0$.
\begin{lemma}\la{my-sol}(\cite[Proposition 5.6]{BS})
  The solution of (\ref{my}), with  initial condition and source term
  satisfying (\ref{myh0})
  and (\ref{mY}) respectively, is bounded as follows for $t>0$:
  \be\la{Ricatti}
     |y(t)-\frac{y(0)}{1+2\rIm K_{T}y(0)t}|\le c\ov Y\Big(\frac{\ve}{1+\ve t}\Big)^{3/2},
     \quad c=c(y_0,\rIm K_{T}).
  \ee
\end{lemma}
Now we consider equation (\ref{mh}) for $h_1$.
\begin{lemma}\la{mh-sol}
 The solution of (\ref{mh}), with  initial condition and source term
  satisfying (\ref{myh0}) and (\ref{mH}), is bounded as follows:
  \be\la{mh-est}
    \Vert h_1\Vert_{L^{\infty}_{-\beta}}\le c(\om_T)\Bigl(\frac{\ve}
    {1+\ve t}\Bigr)^{3/2}\Big(h_{0}+\ov H_1+\sum\limits_m \ov A_m\Big).
  \ee
\end{lemma}
\bp
The function $h_1(x,t)$  can be  expressed as:
 $$
     h_1=e^{\int_0^t{\bf C}_M(\tau)d\tau}h_1(0)
     +\int_0^te^{\int_s^t{\bf C}_M(\tau)d\tau}H(s)ds.
 $$
To establish (\ref{mh-est}) we use the representation (\ref{CM-int}) and
the bounds (\ref{b1}) and (\ref{b5}) to deduce that
  $$
    \Vert h_1\Vert_{L^{\infty}_{-\beta}}\le \frac{c(\om_T)}{(1+t)^{3/2}}
    \Vert h_1(0)\Vert_{{\cal M}_{\beta}}+\int_0^t\frac{c(\om_T)}
    {(1+(t-s))^{3/2}}(\Vert H_1(s)\Vert_{{\cal M}_{\beta}}+
    \Vert A_m(s)\Vert_{{\cal M}_{\beta}})ds
  $$
  $$
    \le c(\om_T)
    \Big[h_0\Big(\frac{\ve}{1+t}\Big)^{3/2}+\int_0^t\frac {ds}
    {(1+(t-s))^{3/2}}\Big(\ov H_1\Big(\frac{\ve}{1+\ve s}\Big)^{3/2}
    +\sum_m \ov A_m\Big(\frac{\ve}{1+\ve s}\Big)^{3/2}\Big)\Big]
  $$
  $$
    \le c(\om_T)\Big(\frac{\ve}
    {1+\ve t}\Big)^{3/2}\Big(h_0+\ov H_1+\sum\limits_m \ov A_{m}\Big),
  $$
since $\int_0^t(1+t-s)^{-3/2}(1+\epsilon s)^{-3/2}ds\leq c(1+\epsilon t)^{-3/2}$
by \cite[lemma 5.3]{BS}.
\ep
\subsection{Inequalities for  majorants}
\label{est-maj}
In this section we estimate in turn the three majorants
$\M_0,\M_1,\M_2$.
\begin{lemma}\la{M0-est}
  The majorants $\M_0(T)$, $\M_1(T)$, and $\M_2(T)$ satisfy
  \be\la{M0M}
     \M_0(T)={\cal R}(\ve^{1/2}\M)\Bigl[(1+\M_1)^4+\ve (1+|\M|)^2\Bigr],
  \ee
  \be\la{M1M}
    \M_1^2={\cal R}(\ve^{1/2}\M)\Big(1+\ve^{1/2}(1+|\M|)^5\Big)
  \ee
  \be\la{M2M}
    \M_2={\cal R}(\ve^{1/2}\M)\Big[(1+\M_1)^3+\ve^{1/2}(1+|\M|)^4\Big].
  \ee
\end{lemma}
\bp
{\it Step i)}
  Using the equality $f=g+h=g+k+k_1+h_1$ and bound (\ref{k1-1}) for $k_1$  we have, using the
notation defined prior to \eqref{om-dif}
  \beqn\nonumber
     |z|^2+\Vert f\Vert_{L^{\infty}_{-\beta}}&=&
     {\cal R}_2(\om,|z|+\Vert f\Vert_{L^{\infty}_{-\beta}})
     (\Vert k_1\Vert_{L^{\infty}_{-\beta}}+|z|^2
     +\Vert h_1\Vert_{L^{\infty}_{-\beta}})\\
     \nonumber
     &=&{\cal R}(\ve^{1/2}\M)\Big(\frac{\ve}{(1+ t)^{3/2}}+
     \Big(\frac{\ve}{1+\ve t}\Big)\M_1^2
     +\Big(\frac{\ve}{1+\ve t}\Big)^{3/2}\M_2\Big)\\
     \nonumber
      &=&{\cal R}(\ve^{1/2}\M)\Big(\frac{\ve}{1+\ve t}\Big)
      \Big(1+\M_1^2+\ve^{1/2}\M_2\Big),
  \eeqn
so that
  \be\la{zf}
     |z|^2+\Vert f\Vert_{L^{\infty}_{-\beta}}={\cal R}(\ve^{1/2}\M)
     \frac {\ve}{1+\ve t}\Bigl(1+\M_1^2+\ve^{1/2}\M_2\Bigr).
  \ee
  Then   (\ref{om-dif}) and (\ref{M0})  imply (\ref{M0M}).
\\
{\it Step ii)}
  Recall  $y=|z_1|^2$ satisfies (\ref{my}) with $Y=Y_R$,
  and  $Y_R$ satisfies the inequality (\ref{YRest}) which is exactly
  the condition (\ref{mY}) with $\ov Y={\cal R}(\ve^{1/2}\M)(1+|\M|)^5$.
  Using (\ref{Ricatti}) as well as (\ref{fy0}) to bound the initial condition $y(0)$,
  it follows that
  $
    y\le {\cal R}(\ve^{1/2}\M)\Big[\frac{\ve}{1+\ve t}+\Big(\frac{\ve}{1+\ve t}\Big)
    ^{3/2}(1+|\M|)^5\Big].
  $
  Therefore
  $$
    |z|^2\le y+ {\cal R}(\om)|z|^3
    \le {\cal R}(\ve^{1/2}\M)\Big[\frac{\ve}{1+\ve t}+\Big(\frac{\ve}{1+\ve t}\Big)
    ^{3/2}(1+|\M|)^5+\Big(\frac{\ve}{1+\ve t}\Big)^{3/2}\M_1^3\Big],
  $$
from which (\ref{M1M}) follows.
\\
{\it Step iii)}
Let us now consider $h_1$, the solution of (\ref{h5}). It has the form (\ref{mh})
  with $H=\hat H_R=\tilde H_R+H'$, where $\tilde H_R$ and $H'$ identify respectively
  to $H_1$ and $H_2$.
  More precisely, using (\ref{fHR}) and (\ref{Am-final}), we have
  $$
   \ov{H_1}={\cal R}(\ve^{1/2}\M)\Big((1+\M_1)^3+\ve^{1/2}(1+|\M|)^4)\Big),\quad\quad
   \ov A_{m}={\cal R}(\ve^{1/2}\M)\Big(\M_1^3+\ve^{1/2}(1+|\M|)^3\Big).
  $$
  We know that $h_1(0)=h(0)$. Thus
  $h_0=c+{\cal R}(\ve^{1/2}\M)\ve^{1/2}\M_0(1+\M_1^2+\ve^{1/2}\M_2)$
  by (\ref{fh0}). Applying Lemma \ref{mh-sol}, we deduce that
  $$
   \Vert h_1\Vert_{L^{\infty}_{-\beta}}={\cal R}(\ve^{1/2}\M)
   \Big(\frac{\ve}{1+\ve t}\Big)^{3/2}\Big[(1+\M_1)^3+\ve^{1/2}(1+|\M|)^4\Big],
  $$
which implies \eqref{M2M}.
\ep

\subsection{Uniform bounds for majorants}
\label{M-bounds}
Now we prove that if $\ve$ is sufficiently small, all the $\M_i$
are bounded uniformly in $T$ and $\ve$.
\begin{lemma}\la{Mb}
  For $\ve$ sufficiently small, there exists a constant $M$ independent of $T$ and $\ve$,
  such that,
  \be\la{M-est}
    |\M(T)|\le M.
  \ee
\end{lemma}
\bp
  Combining the inequalities (\ref{M0M})-(\ref{M2M}) for the $\M_i$, one get a estimate of the form
  $$
   \M^2\le{\cal R}(\ve^{1/2}\M)[(1+\M_1)^8+\ve^{1/2}(1+|\M|)^8]
  $$
  Replacing $\M_1^2$ in the right-hand by its bound (\ref{M1M}),
  we get an inequality in the form
  $$
   \M^2\le{\cal R}(\ve^{1/2}\M)(1+\ve^{1/2}F(\M))
  $$
  where $F(\M)$ is an appropriate function . From this inequality
  it follows that $\M$ is bounded independent of $\ve\ll 1$,
since $\M(0)$ is small,
  and $\M(t)$ is a continuous function of $t$.
  \ep
\bc\la{bounds}
  The function $\om(t)$ has a limit $\om_{+}$ as $t\to\infty$.
Furthermore, the following estimates hold for all $t>0$:
  \be\la{est0}
    |\om_{+}-\om(t)|\le \frac{M\ve}{1+\ve t},~~
|z(t)|\le M\Big(\frac{\ve}{1+\ve t}\Big)^{1/2},~~
\Vert h_1\Vert_{L^{\infty}_{-\beta}}\le M\Big(\frac{\ve}{1+\ve t}\Big)^{3/2},~~
\Vert f\Vert_{L^{\infty}_{-\beta}}\le\frac{ c(M)\ve}{1+\ve t}
  \ee
\ec
\bp
  Since $|\om_T-\om(t)|\le\M_0\ve/(1+\ve t)$,
  then  applying this result to
  $|\om(t_1)-\om(t_2)|$, we see that $\om(t)$ is a Cauchy sequence. It thus has
  a limit, denoted $\om_{+}$ and the first estimate holds.
The next three results follow immediately.
\ep
\subsection{Large time behaviour of  solution}
\label{lt-as}
Here we deduce from corollary \ref{bounds} a theorem
which describe a large time behaviour of the solution. 
Notice that in the decomposition $f=g+h=g+h_1+k+k_1$, a fixed time $T$
has been chosen, and all the components depend on $\om(T)$. From the above proposition,
we know that $\om(t)$ has a limit $\om_{+}$ as $t\to\infty$,
and since the estimates are uniform in $T$ we can reformulate the 
decomposition by choosing $T=\infty$ and $\om_T=\om_+$.
Namely, let us denote ${\bf P}_{\infty}^c={\bf P}^c(\om_{+})$
and ${\bf P}_{\infty}^d=1-{\bf P}_{\infty}^c$. We define $f=g+h$ where
$g={\bf P}_{\infty}^df$ and $h={\bf P}_{\infty}^cf$. We also decompose $h=h_1+k+k_1$ where
$$
  k=a_{20}z^2+a_{11}z\ov z+a_{02}\ov z^2,\quad
  k_1=-\exp\big(\int_0^t {\bf C}_{+}(\tau)d\tau\big)k(0),\quad
  a_{ij}=a_{ij}(\om_{+},x)
$$ 
and
${\bf C}_{+}={\bf C}(\om_{+})+i\big(\om(t)-\om_{+}+\dot\gamma(t)\big)
\big({\bf\Pi}_{\infty}^{+}-{\bf\Pi}_{\infty}^{-}\big)$.
All the estimates previously obtained in \S\ref{om-trans}-\S\ref{maj}
for finite $T$ can be extended to $T=\infty$ and $\om_T=\om_{+}$ without modification.
Thus we have proved the following result:
\begin{theorem}\la{t-main}
  Let the  conditions of Theorem \ref{main} hold. Then,
for $\ve$ sufficiently small, there exist
$C^1$ functions $\omega(t),\gamma(t),z(t)$ as in lemma \ref{BS}, and
constants $\omega_+\in\R$ and $M>0$, such that for all $t\geq 0$:
  \be\la{solform}
     \psi(x,t)=e^{j(\int_0^t\om(s)ds+\gamma(t))}
     \big(\Psi_\om(x)+z(t)u(x,\om)+\ov z(t)u^*(x,\om)+f(x,t)\big),
  \ee
and
   \be\la{est00}
    |\om(t)-\om_{+}|\le M\frac{\ve}{1+\ve t}, \quad
    |z(t)|\le M\Big(\frac{\ve}{1+\ve t}\Big)^{1/2}, \quad
    \Vert f(t)\Vert_{L^{\infty}_{-\beta}}\le M\frac{\ve}{1+\ve t},
  \ee
so that $\om_{+}=\lim\limits_{t\to\infty}\om(t)\in\R$. A corresponding
statement also holds for $t\to -\infty$.
\end{theorem}
\section{Scattering asymptotics}
\setcounter{equation}{0}
\label{sc-as}
In this section we obtain the scattering asymptotics (\ref{sol-as}).
\subsection{Large time behavior of $z(t)$, $\om(t)$ and $\gamma(t)$}
\label{z-beh}
We start with equation (\ref{z6}) for $z_1$, rewritten as
$
  \dot z_1=i\mu z_1+iK_{+}|z_1|^2z_1+\widehat{\widehat Z}_R
$
with $K_{+}=K(\om_{+})$. By (\ref{hhZR}) the inhomogeneous term
$\widehat{\widehat Z}_R$
satisfies the estimate
$$
  |\widehat{\widehat Z}_R|
  ={\cal R}(\ve^{1/2}M)\frac{\ve^2}{(1+\ve t)^{3/2}\sqrt{\ve t}}
   (1+M^4)= O(t^{-2}),\quad t\to\infty.
$$
On the other hand, we have, from (\ref{fy0}) and (\ref{Ricatti}),
$$
y=\frac{y(0)}{1+2\rIm K_{+}y(0)t}+ O(t^{-3/2}),\quad t\to\infty.
$$
Given the estimate  (\ref{est0}) for $|z|$, and obviously the same one for
$|z_1|$, we have
\be\la{z7}
\dot z_1= i\mu z_1+iK_{+}\frac{y(0)}{1+2\rIm K_{+}y(0)t}z_1+Z_1,\quad
Z_1=  O(t^{-2}),\quad t\to\infty.
\ee
Since by assumption in theorem \ref{main} $y(0)=O(\epsilon)$ we can write
$y(0)=\epsilon y_0$ with $y_0=O(1)$.
Let us denote $2\rIm K_{+}y_0=\epsilon k_{+}$, $\delta=\rRe K_{+}/\rIm K_{+}$
so that $\epsilon K_+ y_0=i\epsilon k_+(1-i\delta)/2$.
The solution $z_1$ of (\ref{z7}) is written in the form
$$
 z_1=\!\frac{e^{i\int_0^t\mu(t_1)dt_1}}
{(1+{\epsilon}k_{+}t)^{\frac{1}{2}(1-i\delta)}}
 \Big[z_1(0)+\!\int_0^t\!\!e^{-i\int\limits_0^s\mu(t_1)dt_1}
 (1+{\epsilon}k_{+}s)^{\frac{1}{2}(1-i\delta)}Z_1(s)ds\Big]
  =z_{\infty}\frac{e^{i\int_0^t\mu(t_1)dt_1}}{(1+{\epsilon}k_{+}t)^{\frac{1}{2}(1-i\delta)}}
  +z_R
$$
where
$$
 z_{\infty}(\om)= z_1(0)+\int_0^{\infty}e^{-i\int_0^s\mu(t_1)dt_1}
 (1+{\epsilon}k_{+}s)^{\frac{1}{2}(1-i\delta)}Z_1(s)ds
$$
and
$$
 z_R= -\int_t^{\infty}e^{i\int_s^t\mu(t_1)dt_1}
 \Big(\frac{1+{\epsilon}k_{+}s}{1+{\epsilon}k_{+}t}
\Big)^{\frac{1}{2}(1-i\delta)}Z_1(s)ds.
$$
Here  $\mu(t_1)=\mu(\om(t_1))$. From the bound (\ref{z7}) on $Z_1$ it follows that
$z_R = O(t^{-1})$.
Therefore $z(t)=z_1(t)+ O(t^{-1})$  satisfies
$$
z(t)= z_{+}\frac{e^{i\int_0^t\mu(t_1)dt_1}}
{(1+{\epsilon}k_{+}t)^{\frac{1}{2}(1-i\delta)}}
 + O(t^{-1}),\;t\to\infty,\quad z_{+}=z_{\infty}(\om_{+}).
$$
From these formulas for $z(t)$,  the
asymptotic behavior of $\om(t)$ and $\gamma(t)$ can be
deduced as in \cite[Sections 6.1 and 6.2]{BS}, leading to the following:
\begin{lemma}\label{om-as}
In the situation of  theorem \ref{t-main}, the functions $\om(t)$
and $\gamma(t)$ have the following asymptotic behavior as $t\to +\infty$:
  \beqn\nonumber
    \om(t)& = &\om_{+}+\frac{q_{+}}{1+{\epsilon}k_{+}t} +\frac{b_{+}}{1+{\epsilon}k_{+} t}
    \cos(2\mu_{+}t+b_1\log(1+{\epsilon}k_{+}t)+b_2)+  O(t^{-3/2}),
  \\
\nonumber
\gamma(t)& = &\gamma_{+}+c_{+}\log(1+{\epsilon}k_{+}t)+ O(t^{-1}),
\eeqn
where $\omega_+,k_+$ are as defined above,
$\mu_{+}=\mu(\om_{+})$, and
$q_+$, $b_{+}$, $b_1$, $b_2$, $c_+$ are constants.
\end{lemma}
\subsection{Soliton asymptotics}
\label{solas-sec}

Here we prove the statement \eqref{sol-as} in
our main theorem \ref{main}.
To achieve this
we look for the solution $\psi(x,t)$ to (\ref{S}), in the
corresponding complex form $\psi=s+\rv+f$,
where
$$
s(x,t)=\psi_{\om(t)}(x)e^{i\theta(t)},\quad\dot\theta(t)=\om(t)+\dot\gamma(t)
$$
is the accompanying soliton, and
$$
\rv(x,t)=v(x,t)e^{i\theta(t)},\quad
v(x,t)=\big(z(t)+\ov z(t)\big)u_1(x,\om(t))+i\big(z(t)-\ov z(t)\big)u_2(x,\om(t)).
$$
We aim now to prove the complex form \eqref{sol-as'} of the
scattering asymptotics, by analysing $f.$
\begin{lemma}
If $\psi(x,t)$ is a solution of (\ref{S}), then
$i\dot f=-f''+R$, where
\beqn\la{F}
R&=&\dot\gamma (s+\rv)-i\dot\om\partial_{\om}(s+\rv)
-i[(\dot z-i\mu z)(u_i+iu_2)+(\dot{\ov z}+i\mu\ov z)(u_1-iu_2)]e^{i\theta}
\nonumber\\
&&-\delta(x)[af+be^{i\theta}\rRe (e^{-i\theta}f)
+\cO(|f+\rv|^2)].
\eeqn
\end{lemma}
\begin{proof}
Multiply (\ref{NEP}) by $e^{i\theta}$, we obtain
$-\om s=-s''-\delta(x)F(s)$ that implies
\be\la{dots}
i\dot s=-\om s-\dot\gamma s+i\dot\om\partial_{\om}s=
-s''-\dot\gamma s+i\dot\om\partial_{\om}s-\delta(x)F(s).
\ee
By the equation $\bC u=i\mu u$,
we obtain for the components $u_1$ and $u_2$ of vector $u$,
$$
 -u_2''+\om u_2-\delta(x)au_2=i\mu u_1,\quad~~
-u_1''+\om u_1-\delta(x)[a+b]u_1=-i\mu u_2.
$$
Therefore
$-v''+\om v-\delta(x)[av+b\rRe v]=-\mu (z-\ov z)u_1-i\mu(z+\ov z)u_2$
and then
\beqn\nonumber
\!\!\!\!\!&&\!\!\!\!\!\!\!\!\!\!\!\!\!\!\!\!
i\dot\rv=-(\om +\dot\gamma) ve^{i\theta(t)}+i\dot\om\partial_{\om}ve^{i\theta(t)}
+\big(i(\dot z+\dot{\ov z})u_1-(\dot z-\dot{\ov z})u_2\big)e^{i\theta}\\
\nonumber
&=&\!\!\!-\rv''-\dot\gamma \rv+i\dot\om\partial_{\om}\rv
-\delta(x)[av+b\rRe v]e^{i\theta}+i[(\dot z-i\mu z)(u_1+iu_2)
+(\dot{\ov z}+i\mu\ov z)(u_1-iu_2)]e^{i\theta}
\eeqn
The last equality, (\ref{S}) and (\ref{dots}) imply 
\be\la{F-def}
i\dot f= -f''+R
\ee
with $R$ as in \eqref{F}.
\end{proof}
The function $f(x,t)$ which is a solution of (\ref{F-def}) can be expressed  as
\beqn\nonumber
f(t)&= &W(t)f(0)+\int_0^t W(t-\tau)R(\tau)
d\tau\\
&= &W(t)\Big(f(0)+
\int_0^{\infty} W(-\tau)R(\tau)d\tau\Big)
\nonumber
-\int^{\infty}_t W(t-\tau)R(\tau)d\tau
= W(t)\phi_{+}+r_{+}(t)
\eeqn
where $W(t)$ is the dynamical group of the free Schr\"odinger equation.
To establish the asymptotic behavior (\ref{sol-as'}), it suffices to
prove that
\be\la{phi-r}
\phi_{+}\in C_b(\R)\cap L^2(\R),\quad {\rm and}\quad
\Vert r_{+}(t)\Vert_{C_b(\R)\cap L^2(\R)}= O(t^{-\nu}),\;t\to\infty.
\ee
These assertions follow from the definition (\ref{F}) of the
function $R$, and the following two lemmas.
The first lemma studies the contribution
to $\phi_{+}(x)$ and $r_{+}(x,t)$ from the
terms  in (\ref{F}) involving $\delta(x)$: these are  
${\cal O}(t^{-1})$ as $t\to\infty$  by corollary \ref{bounds}. 
\begin{lemma}\la{L2-1}
Let $\Pi(t)$ be a continuous bounded function of $t\ge 0$,
with $|\Pi(t)|\le L_0$ and $|t\Pi(t)|\le L_1$. Then
\be\la{Cb}
\int_0^{\infty} W(-\tau)[\delta(\cdot)\Pi(\tau)]d\tau=
\int_0^{\infty}\frac {e^{-ix^2/(4\tau)}}
{\sqrt{-4\pi i\tau}}\Pi(\tau)d\tau\in C_b(\R)\cap L^2(\R)
\ee
and for  $\nu\in(0,\frac{1}{4})$ there exists $C=C(\nu,L_0,L_1)>0$ such that
\be\la{Cbr}
\Big\|\int_t^{\infty} W(t-\tau)[\delta(\cdot)\Pi(\tau)]d\tau\Big\|_{C_b(\R)\cap L^2(\R)}
\le C(1+t)^{-\nu}
\ee
\end{lemma}
\begin{proof}
The $C_b$-properties follow from formulas
(\ref{Cb}) and (\ref{Cbr}) (in fact with $\nu=1/2$).
To prove the $L^2$-properties, let us  change the variable to $\tau=1/u$ to get:
\be\la{ftz}
\psi(x):=\frac 1{\sqrt{-4\pi i}}\int_0^\infty e^{-iux^2/4}~\eta(u)~du
=\frac 1{\sqrt{-2i}}{\mathcal F}_{u\to x^2/4}(\theta(u)\eta(u)),
\qquad \eta(u)=\Pi(1/u)/u^{3/2},
\ee
where
${\mathcal F}_{u\to\xi}(f(u))=\hat f(\xi)$
indicates the Fourier transform
with argument $\xi$, and $\theta(u)$ is the Heaviside function.
By the assumptions on $\Pi$ we have $|\eta(u)|\leq  L_0u^{-3/2}$ as
$u\to \infty$, and
$|\eta(u)|\leq L_1u^{-1/2}$ as $u\to 0$.
Therefore $\eta(u)\in L^p(\R)$ for $1\le p<2$.
It follows from the Hausdorff-Young inequality for the Fourier
transform that
$\psi\in L^q(\R)$ for $q>2$ {\em as a function of $y=x^2$}, i.e.
$$
\int_0^\infty|\psi(x)|^qx~dx<\infty ,\;\forall q>2,
$$
and hence $\psi\in L^2(\R)$, since it is already known to be
bounded and continuous.
It remains to prove the decay (\re{Cbr})
in the norm $L^2(\R)$. Denote
$$
\rho(x,t):=\int_t^{\infty}\! W(-\tau)[\delta(\cdot)\Pi(\tau)]d\tau
=\frac 1{\sqrt{-4\pi i}}\int_0^{1/t}e^{-iux^2/4}~\eta(u)~du
=\frac 1{\sqrt{-2i}}
{\mathcal F}_{u\to x^2/4} (\zeta_t(u)\eta(u)).
$$
Here $\zeta_t(u)$ is the characteristic function of the interval $(0,1/t)$.
As above $\rho $ is bounded, but also since
$$
\Vert \zeta_t(u)\eta(u)\Vert_{L^p}=
\Big(\int_0^{1/t}|\eta(u)|^pdu\Big)^{1/p}
\le L_1\Big(\int_0^{1/t}u^{-p/2}du\Big)^{1/p}
\le Ct^{-\frac{1-p/2}p},~~~ 1\le p<2,~~~ t\ge 1
$$
the Hausdorff-Young inequality implies that
for any $q>2$:
$$\Vert\rho(x,t)(1+|x|)^{1/q}\Vert_{L^q}\leq Ct^{-\frac{1-p/2}p},\quad t\ge 1$$
for some constant $C=C(L_1,p)$, for $q^{-1}+p^{-1}=1$.
The Young inequality then implies that
$$
\Vert\rho(x,t)\Vert_{L^2}\le \Vert\rho(x,t)(1+|x|)^{1/q}\Vert_{L^{q}}
 \Vert (1+|x|)^{-1/q}\Vert_{L^{r}}\le Ct^{-\frac{1-p/2}p},\;
 q^{-1}+r^{-1}=1/2
$$
if $r>q$.
To have  $r>q$, we must take $q<4$, or
equivalently $p>4/3$. Hence, we have $\nu=(1-p/2)/p<1/4$.
\end{proof}
The second lemma studies
the contribution to $\phi_{+}(x)$ and $r_{+}(x,t)$ from terms
without  $\delta(x)$ in (\ref{F}). Consider the expansions  (\ref{omega2}),
(\ref{gamma2}), (\ref{z2}), for $\dot\om(t)$, $\dot\gamma(t)$,
and $\dot z(t)-i\mu z(t)$: the main (quadratic) parts of these contain the terms
$z^2_{+}(t)$, $\ov z^2_{+}(t)$, $z_{+}(t)\ov z_{+}(t)$, which are $ O(t^{-1})$
as $t\to\infty$. The remainders are $O(t^{-3/2})$, and it is straightforward
 to bound the contribution of these to $\phi_+$
in $C_b\cap L^2$, and to check that these contribute $O(t^{-1/2})$ to $r_+$ in
$C_b\cap L^2$. Further, $\dot\theta(t)-\om_+(t)=O(\tau^{-1})$.
Thus, without loss of generality, we may replace $\dot\om(t)$,
$\dot\gamma(t)$, and $\dot z(t)-i\mu z(t)$  by the main quadratic parts
and treat these terms with the phase $\theta(t)$ in \eqref{F} replaced by
$\varphi_+(t)\equiv\om_{+}t$.
\begin{lemma}\la{L2-2}
Let $\Pi(t)$ be one of the functions $z^2_{+}(t)e^{i\om_{+}t}$,
$|z_{+}(t)|^2e^{i\om_{+}t}$ or $\ov z_{+}^2(t)e^{i\om_{+}t}$, where
$z_{+}(t)=e^{i\mu_{+}t}/(1+t)^{1/2}$,
and let  $\psi(x)\in L^2(\R)\cap L^1(\R)$. Then
\be\la{WPi}
  \int_0^{\infty}\Pi(\tau) W(-\tau)\psi d\tau\in C_b(\R)\cap L^2(\R)
\ee
and
\be\la{WPii}
  \bigl\|\int^{\infty}_t\Pi(\tau) W(t-\tau)\psi d\tau\bigr\|_{C_b(\R)\cap L^2(\R)}
  \leq C t^{-1/3},\quad t>1.
\ee
\end{lemma}
\begin{proof}
Since $\Vert W(t)\psi\Vert_{C_b}= O(t^{-1/2})$ then  $C_b$-- properties are evident.
It remains to prove the $L^2$-properties.
First, we consider the case when $\Pi(t)=|z_{+}(t)|^2e^{i\om_{+}t}$.
We have
\be\la{It}
\Big\Vert\int_t^{\infty}
\frac{e^{i(\xi^2+\om_+)\tau}\hat\psi(\xi)d\tau}{1+\tau}\Big\Vert_{L^2}\le \frac C{1+t}
\Vert \hat\psi(\xi)/(\xi^2+\om_+)\Vert_{L^2}
\ee
since the partial integration implies that
$$
|\int_t^{\infty}\fr{e^{i(\xi^2+\om_+)\tau}}{1+\tau}~d\tau|
\le\Big|\fr{e^{i(\xi^2+\om_+)\tau}}{(\xi^2+\om_+)(1+t)}\Big|
+\Big|\int_t^{\infty}\fr{e^{i(\xi^2+\om_+)\tau}}
{(\xi^2+\om_+)(1+\tau)^2}~d\tau\Big|
\le\frac{C}{(\xi^2+\om_+)(1+t)}
$$
Hence (\ref{WPi}) and (\ref{WPii}) for $\Pi(t)=|z_{+}(t)|^2e^{i\om_{+}t}$ follow.
For $\Pi(t)=z_{+}^2(t)e^{i\om_{+}t}$ the proof is similar.
Finally, we consider $\Pi(t)=\ov z_{+}^2(t)e^{i\om_{+}t}$.
It suffices to prove that
\be\la{Pt}
I(t)=\Big\Vert\int_t^{\infty}\fr{e^{i(\xi^2+\om_+-2\mu_+)\tau}
\hat\psi(\xi)~d\tau}{1+\tau}\Big\Vert_{L^2}={\cal O}(t^{-1/3})
\ee
For the fixed $0<\beta<1$ let us define
$$
\chi_{\tau}(\xi)=\left\{\ba{ll}1,~{\rm if}~|\xi-\sqrt{2\mu_+-\om_+}|\le 1/\tau^{\beta}~
{\rm or}~|\xi+\sqrt{2\mu_+-\om_+}|\le 1/\tau^{\beta}
\\ 0,~~|\xi\pm\sqrt{2\mu_+-\om_+}|> 1/\tau^{\beta} \ea\right.
$$
 Then
$$
I(t)\!\le\!\Big\Vert\!\int_t^{\infty}\!\!
\fr{e^{i(\xi^2+\om_+-2\mu_+)\tau}\chi_{\tau}(\xi)\hat\psi(\xi)d\tau}
{1+\tau}\Big\Vert_{L^2}\!
+\Big\Vert\!\int_t^{\infty}\!\!
\fr{e^{i(\xi^2+\om_+-2\mu_+)\tau}(1\!-\!\chi_{\tau}(\xi))\hat\psi(\xi)d\tau}
{1+\tau}\Big\Vert_{L^2}\!=I_1(t)+I_2(t)
$$
Since $\hat\psi(\xi)$ is bounded function,
and $\Vert\chi_\tau\Vert^2_{L^2}\le 4/\tau^\beta$, we have
$I_1(t)\le C\Vert\hat\psi\Vert_{L^\infty}/(1+t)^{\beta/2}$.\\
On the other hand,
the partial integration implies that
$$
I_2(t)=\Big\Vert\int_t^{\infty}\!
\fr{(1-\chi_{\tau}(\xi))\hat\psi(\xi)~de^{i(\xi^2+\om_+-2\mu_+)\tau}}
{(\xi^2+\om_+-2\mu_+)(1+\tau)}\Big\Vert_{L^2}\!\le
\fr{Ct^{\beta}}{1+t}\Vert\hat\psi\Vert_{L^2}
+C\int_t^{\infty}\!\!\frac{\tau^{\beta}d\tau}
{(1+\tau)^2}\Vert\hat\psi\Vert_{L^2}\!
\le \fr{C\Vert\hat\psi\Vert_{L^2}}{(1+t)^{1-\beta}}
$$
Equating $\beta/2=1-\beta$, we get $\beta=2/3$.
Hence (\ref{Pt})  follows.
\end{proof}
\br\la{brt}
The $t\to -\infty$ case is handled in an identical way.
\er
\setcounter{section}{0}
\setcounter{equation}{0}
\protect\renewcommand{\thesection}{\Alph{section}}
\protect\renewcommand{\theequation}{\thesection. \arabic{equation}}
\protect\renewcommand{\thesubsection}{\thesection. \arabic{subsection}}
\protect\renewcommand{\thetheorem}{\Alph{section}.\arabic{theorem}}

\section{Eigenfunctions of  discrete spectrum}
\setcounter{equation}{0}
\label{eig-fd}
Here we  find the function $u=u(\om)$ satisfying
${\bf C}u=\lam u$, where $\lam=i\mu$.
Using the definition  of the operator ${\bf C}$, we obtain
\be\la{u-eqn}
  \left(
   \ba{rcr}-\lam              &&-\Delta+\om\\
   \Delta-\om    &&-\lam \ea \right)u =
   \de(x)
   \left(
   \!\ba{cc}
   0      &a\\
   -a-b &0
   \ea\!\right)u,\quad \Delta=\ds\fr{d^2}{dx^2}.
\ee
If $x\ne 0$, the equation (\ref{u-eqn}) takes the form
\be\la{hom}
   \left(
   \ba{rcr}-\lam              &&-\Delta+\om\\
   \Delta-\om    &&-\lam \ea \right)u =
   0,~~~~~x\ne 0.
\ee
General solution is a linear combination of exponential solutions of type $e^{ikx}v$.
Substituting to (\re{hom}), we get
\be\la{homh}
   \left(
   \ba{rcr}-\lam   &&  k^2+\om\\
   -k^2-\om        &&    -\lam
   \ea
   \right)v=0.
\ee
For nonzero vectors $v$, the determinant of the matrix vanishes:
$\lam^2+(k^2+\om)^2=0$. Then $k_\pm^2+\om=\mp i\lam$.
Finally, we obtain four roots $\pm k_\pm(\lam)$,
$k_\pm(\lam)=\sqrt{-\om\mp i\lam}$,
where the square root is defined as an analytic continuation
from a neighborhood of the zero point $\lam=0$
taking the positive value of $\rIm\sqrt{-\om}$ at $\lam=0$.
We choose the cuts ${\cal C}_+$ in the complex plane $\lam$ from the branching points to
infinity. Then
$\rIm k_\pm(\lam)>0$ for $\lam\in\C\setminus {\cal C}_\pm$
and we have two corresponding vectors $v_\pm=(1,\pm i)$
and  four linearly independent exponential solutions
$$
   v_+e^{\pm ik_+ x}=(1,\, i)e^{\pm ik_+ x},~~~~~~~~~~~~
   v_-e^{\pm ik_- x}=(1, -i)e^{\pm ik_- x}.
$$
Now we find the solution to (\ref{u-eqn}) in the form
\be\la{u-sol}
  u=Ae^{ik_+ |x|}v_{+}+Be^{ik_- |x|}v_{-}.
\ee
At the point $x=0$ we have a jump:
\be\la{jump}
 u'(+0)-u'(-0)=-\left(
 \ba{cc}
 a+b&0\\
 0&a
\ea\right)u(-0)
\ee
Substituting (\ref{u-sol}) into (\ref{jump}), we obtain
\be\la{j0}
2ik_{+}Av_{+}+2ik_{-}Bv_{-}=-M(Av_{+}+Bv_{-}),\quad
M=\left(\ba{cc}a+b&0\\0&a\ea\right).
\ee
Note that
$$
Mv_+=\al v_++\beta v_-,\quad
Mv_-=\al v_-+\beta v_+,~~~~~~{\rm where}~~~~
\al=a+\fr{b}2,\quad\beta=\fr{b}2.
$$
Then (\re{j0}) becomes
$$
\left\{
\ba{l}
(2ik_{+}+\al)A+\beta B=0
\\
\\
\beta A+(2ik_{-}+\al)B=0
\ea\right.
$$
The determinant $D=(2ik_++\al)(2ik_-+\al)-\beta^2$ vanishes for $\lam=i\mu$
since $i\mu$ lies in the spectrum.
Therefore, we set $A=1$, and obtain  the corresponding eigenfunction:
\be\la{u-def}
  u=e^{ik_{+}|x|}v_{+}-\frac{\beta}{2ik_{-}+\al}e^{ik_{-}|x|}v_{-}.
\ee
We have $2ik_{-}+\al\not=0$. Indeed, if  $2ik_{-}+\al=0$, then $\beta=0$, $\al=a$,
and $2ik_{-}+\al=-2\sqrt{\om+\mu}+a=-\sqrt{a^2+4\mu}+a<0$. \\
Since both $k_{+}=\sqrt{-\om+\mu}$ and $k_{-}=\sqrt{-\om-\mu}$ are purely imaginary,
the first component
$u_1$ is real, while the second one $u_2$ is imaginary.
It is easy to prove that $u^{*}=(u_1,-u_2)$ is the eigenfunction associated to $\lam=-i\mu$.

\section{Eigenfunctions of continuous spectrum}
\setcounter{equation}{0}
\label{eig-fc}
Let $\lam=i\nu$ with some $\nu>\om$.
First, we find an even solution $u=\tau_+$ to
(\ref{u-eqn}) in the form
$$
  \tau_+=(Ae^{ik_+ |x|}+Be^{-ik_+ |x|})v_{+}+Ce^{ik_- |x|}v_{-}
$$
Similarly (\ref{jump}) and (\ref{j0}), we obtain
$$
\left\{
\ba{l}
2ik_{+}(A-B)=-\al(A+B)-\beta C
\\
\\
2ik_{-}C=-\beta (A+B)-\al C
\ea\right.
$$
If $2ik_{-}+\al=0$, then $D=-\beta^2\not=0$. We put $A=D$, then
$B=-D$, $C=4\beta ik_{+}$.\\
If $2ik_{-}+\al\not=0$, we put $A=\ov D$ and obtain 
$B=-D$, and  $C=4\beta ik_{+}$.
Finally,
\be\la{tau-sol1}
 \tau_+=(\ov De^{ik_+ |x|}-De^{-ik_+ |x|})v_{+}+4\beta ik_{+}e^{ik_- |x|}v_{-}.
\ee
It is easily to check that  an odd solution $u=s_+$ to (\ref{u-eqn}) is
\be\la{s-sol}
  s_+=\frac 1{2i}(e^{ik_{+}x}-e^{-ik_{+}x})v_{+}=\sin (k_{+}x)v_{+}.
\ee
In the case $\nu<-\om$ we obtain similarly that
$$
\tau_-=(\ov De^{ik_{-}|x|}-De^{-ik_{-}|x|})v_{-}+4\beta ik_{-}e^{ik_{+}|x|}v_{+},
\quad s_-= \sin (k_{-}x)v_{-}.
$$
\section{Proof of Proposition \ref{l1}}
\setcounter{equation}{0}
\label{prop-pr}
Denote by  $B$  a Banach space with the norm $\Vert\cdot\Vert$.
\begin{lemma}\label{ek}(cf.\cite[lemma B.1]{mk})
  Let the operator $K(t)$, $t>0$, satisfies
  \be\la{K0}
    K(t)=\int\zeta(\nu)e^{i\nu t}Q(\nu)d\nu,\quad Q(\nu):=\frac{L(\nu)-L(\nu_0)}
    {\nu-\nu_0},
  \ee
where $\zeta\in C_0^\infty(\R)$, $\zeta(\nu)=1$ in the some vicinity of $\nu_0$,
and for $k=0,1,2$
\be\la{f-dif}
 M_k:= \sup_{\nu\in\supp\!\zeta}
\Vert\partial ^k_\nu L(\nu)\Vert<\infty.
  \ee
Then
  \be\la{nK}
    \Vert K(t)\Vert={\cal O}(t^{-3/2}),\quad t\to\infty,
  \ee
\end{lemma}
\noindent
{\it Proof of Proposition \ref{l1}}. We will use the Laplace representation:
  $$
    e^{{\bf C}t}({\bf C}-2i\mu-0)^{-1}=-\frac 1{2\pi i}\int_{-i\infty}^{i\infty}
    e^{\lam t}R(\lam+0)~d\lam ~ R(2i\mu+0).
  $$
  Let us apply the Hilbert identity for the resolvent:
  $$
    R(\lam_1)R(\lam_2)=\frac 1{\lam_1-\lam_2}[R(\lam_1)-R(\lam_2)],\;\rRe\lam_k>0,
    \; k=1,2
  $$
  to $\lam_1=\lam+0$ and $\lam_2=2i\mu+0$. Then we obtain
  $$
    e^{{\bf C}t}({\bf C}-2i\mu-0)^{-1}=-\frac 1{2\pi i}\int_{-i\infty}^{i\infty}
    e^{\lam t}\frac{R(\lam+0)- R(2i\mu+0)}{\lam-2i\mu}~d\lam
  $$
  $$
    =-\frac 1{2\pi i}\!\int_{-i\infty}^{i\infty}\!
    e^{\lam t}\zeta(\lam)\frac{R(\lam+0)- R(2i\mu+0)}{\lam-2i\mu}~d\lam
    -\frac 1{2\pi i}\!\!\int\limits_{{\cal C}_+\cup{\cal C}_-}\!\!
    e^{\lam t}(1-\zeta(\lam))\frac{R(\lam+0)- R(2i\mu+0)}{\lam-2i\mu}~d\lam
  $$
  $$
    -\frac 1{2\pi i}\int\limits_{(-i\infty,i\infty)\setminus({\cal C}_+\cup{\cal C}_-)}
    e^{\lam t}(1-\zeta(\lam))\frac{R(\lam+0)- R(2i\mu+0)}{\lam-2i\mu}~d\lam
    ={\bf K}_1(t)+{\bf K}_2(t)+{\bf K}_3(t)
  $$
  where $\zeta(\lam)\in C_0^{\infty}(i\R)$,
  $\zeta(\lam)=1$ for $|\lam-2i\mu|<\delta/2$ and $\zeta(\lam)=0$ for $|\lam-2i\mu|>\delta$,
  with $0<\delta<2\mu-\om$.
  By Lemma \ref{ek} with $L(\nu)=R(i\nu+0)$ and $B={\cal B}_{\beta}$
  with $\beta\ge 2$, we obtain that
  $$
    \Vert {\bf K}_1(t)\Vert_{{\cal L}({\cal M}_{\beta},L^{\infty}_{-\beta})}
    ={\cal O}({t^{-3/2}}),\quad t\to\infty.
  $$
  since in this case the bounds (\ref{f-dif}) follow from formulas (\ref{nR}).
  Further, we can apply the arguments from the proof
  of proposition \ref{TD} in \cite{BKKS}  and obtain
  $$
    \Vert {\bf K}_2(t)\Vert_{{\cal L}({\cal M}_{\beta},L^{\infty}_{-\beta})}
    ={\cal O}({t^{-3/2}}),\quad t\to\infty.
  $$
  Here the choice of the sigh in ${\bf C}-2i\mu-0$ plays the crucial role.\\
  Finally, the integrand in ${\bf K}_3(t)$ is an analytic function of
  $\lam\not =0, \pm i\mu$ with the values in 
  ${\cal L}({\cal M}_{\beta},L^{\infty}_{-\beta})$
  for $\beta\ge 0$. At the points $\lam=0$ and $\lam=\pm i\mu$ the
  integrand has the poles of finite order. Hoverever, all the Laurent
  coefficients vanish when applied to ${\bf P}^{c}h$. Hence for ${\bf K}_3(t)$
  we obtain, twice integrating by parts,
  $$
    \Vert {\bf K}_3(t){\bf P}^ch\Vert_{L^{\infty}_{-\beta}}
    \le c(1+t)^{-3/2}\Vert h\Vert_{{\cal M}_{\beta}},
  $$
completing the proof. \qed
\section{Proof of Lemma \ref{rem2}}
\setcounter{equation}{0}
\label{lemma-pr}
  We will use the following representation (see \cite{BKKS}):
  \be\la{Riess}
   {\bf P}^c=\frac 1{2\pi i}\int_{{\cal C}_{+}}\!({\bf R}(\lam+0)
   -{\bf R}(\lam-0))d\lam
   +\frac 1{2\pi i}\int_{{\cal C}_{-}}\!({\bf R}(\lam+0)\!-\!{\bf R}(\lam-0))d\lam=
   {\bf\Pi}^++{\bf\Pi}^-.
  \ee
   Let us decompose the resolvent, as given in (\ref{nR}) and the subsequent formula, as
  \be\la{dec}
   {\bf R}(\lam,x,y)={\Gamma}(\lam,x,y)+P(\lam,x,y)=
   \sum\limits_{k=1}^6A_k(\lam,x,y)\tau_k
  \ee
 where
\beqn\nonumber
A_1\!\!&=&\!\!\frac {e^{ik_+|x-y|}-e^{ik_+(|x|+|y|)}}{4k_+},\quad
  A_2=\frac {i\al-2k_-}{2D}e^{ik_+(|x|+|y|)},\quad
  A_3=\frac {i\beta}{2D} e^{ik_+|x|}e^{ik_-|y|}\\\\
\la{AA}\nonumber
A_4\!\!&=&\!\!-\frac {i\beta}{2D} e^{ik_-|x|}e^{ik_+|y|},\quad
   A_5=\frac {-i\al+2k_+}{2D}e^{ik_-(|x|+|y|)},\quad
   A_6=-\frac {e^{ik_-|x-y|}-e^{ik_-(|x|+|y|)}}{4k_-}
\eeqn
and
\be\la{tau1}
\tau_1=\tau_2= \left( \ba{cc}
1   &   -i\\
i   &    1 \ea\right),\; \tau_3= \left( \ba{cc}
1   &   i\\
i   &    -1 \ea\right),\; \tau_4= \left( \ba{cc}
1   &   -i\\
-i   &   -1 \ea\right),\; \tau_5=\tau_6= \left( \ba{cc}
1   &   i\\
-i   &    1 \ea\right) \ee
We have
\beqn\nonumber
P^c_Tj^{-1}-i(\Pi^+_T-\Pi^-_T)\!\!&=&\!\!\frac 1{2\pi i}\int_{{\cal C}_+}
[2(A_3(\lam+0)-A_3(\lam-0))\tau_3+
2(A_5(\lam+0)-A_5(\lam-0))\tau_5]~d\lam\\
\nonumber
\!\!&+&\!\!\frac 1{2\pi i}\int_{{\cal C}_-}
[2(A_2(\lam+0)-A_2(\lam-0))\tau_2+
2(A_4(\lam+0)-A_4(\lam-0))\tau_4]~d\lam.
\eeqn
Let us consider only the integral over ${\cal C}_+$;
the integral over ${\cal C}_-$ can be
dealt with by an identical argument. For $\lam\in{\cal C}_+$ we have:
$k_+$ is real, and $k_+(\lam+0)=-k_+(\lam-0)$
while $k_-$ is pure imaginary with
$\rIm k_->0$ and $k_-(\lam+0)=k_-(\lam-0)$.\\
Since $A_5(\lam,x,y)$ for $\lam\in{\cal C}_+$ exponentially
decay if $|x|,\,|y|\to\infty$ and smallest exponential rate of the
decaying is equal to $(2\omega)^{1/2}$, then the bound (\ref{Proj})
for the integral over ${\cal C}_+$ with integrand $A_5(\lam+0)-A_5(\lam-0)$
is evident.
It remains to consider the integral with $A_3(\lam+0)-A_3(\lam-0)$.
We change variable:
$\zeta=\sqrt{-\omega-i\lam}$ for the first summand and
$\zeta=-\sqrt{-\omega-i\lam}$ for the second. Then we get
$$
\int_{{\cal C}_+}(A_3(\lam+0)-A_3(\lam-0))~d\lam
=-\beta\int_{-\infty}^{+\infty}\frac{e^{-\sqrt{2\omega+\zeta^2}|y|+i\zeta |x|}}
{D(\zeta)}~\zeta d\zeta
$$
where
$$
D(\zeta)=\al^2-\beta^2 +2i\al\zeta-2\al\sqrt{2\omega+\zeta^2}
-4i\zeta\sqrt{2\omega+\zeta^2}.
$$ 
Writing
$e^{i\zeta|x|}~d\zeta=de^{i\zeta|x|}/i|x|$, and
integrating by parts, we obtain that
$$
|\int_{-\infty}^{+\infty}\frac{e^{-\sqrt{2\omega+\zeta^2}|y|+i\zeta |x|}}
{D(\zeta)}~\zeta d\zeta|\le Ce^{-\sqrt{2\omega}|y|}(1+|x|)^{-m},\;\forall m\in {\bf N}
$$
completing the proof of lemma \ref{rem2}.
\section{Fermi Golden Rule}
\setcounter{equation}{0}
\label{Fermi}
In this section we show that  condition (\ref{FGR}) holds generically in a certain
sense: in particular, if $a(\cdot)$ is a polynomial function then, generically,
the set of values of $C$ for which (\ref{FGR}) fails is isolated. 
By (\ref{tau-sol1})
\be\la{t}
\tau_{+}(2i\mu)\mid_{x=0}=(\ov D-D)v_{+}+4\beta ik_{+}v_{-}
 =-4ik_{+}(\al+2ik_{-})v_{+}+4\beta ik_{+}v_{-}=\sigma(\kappa v_{+}+v_{-})
\ee
where $\sigma=4\beta ik_+$, $\kappa=-(\al+2ik_{-})/\beta$,
$k_{\pm}=\sqrt{-\om\pm 2\mu}$.
Further, $E_2[u,u]=\delta(x)\tilde E_2[u(0),u(0)]$, where
$$
\tilde
E_2[u(0),u(0)]=a'(C^2)(u(0),u(0))\Psi(0)+2a''(C^2)(\Psi(0),u(0))^2\Psi(0)
+2a'(C^2)(\Psi(0),u(0))u(0)
$$
By (\ref{u-def}),
$u(0)=\rho v_{+}+v_{-}$, where $\rho=-(2ik_{-}+\al)/\beta$, $k_-=\sqrt{-\om-\mu}$.
Therefore,
$(u(0),u(0))=(\rho+1)^2-(\rho-1)^2=4\rho$ and then
$$
\tilde E_2[u(0),u(0)]=a'(C^2)4\rho\left(\ba{c} C\\0\ea\right)
+2a''(C^2)C^2(\rho+1)^2\left(\ba{c} C\\0
\ea\right)
+2a'(C^2)C(\rho+1)\left(\ba{c} \rho+1\\i(\rho-1)
\ea\right)
$$
The last equality and  (\ref{t}) imply
$$
\langle\tau_{+}(2i\mu), E_2[u,u]\rangle=\sigma (\kappa+1)
\Big[a'4\rho C+2a''C^3(\rho+1)^2+2a'C(\rho+1)^2\Big] +\sigma
(\kappa-1)2a'C(\rho^2-1)
$$
Hence, (\ref{FGR}) is equivalent to the condition
$a''\not=-\ds\frac{2a'(2\kappa\rho+2\rho+\kappa\rho^2+1)}{C^2(\kappa+1)(1+\rho)^2}$
with
$$
\kappa=-\frac{a+\frac b2-2\sqrt{\om+2\mu}}{b/2},\quad
\rho=-\frac{a+\frac b2-2\sqrt{\om+\mu}}{b/2},\quad
\om=\frac{a^2}4,\quad
 \mu=\frac b4\sqrt{a^2-\frac{b^2}4}
$$
Set $\nu =\ds\frac{b}{2a}$. Then 
$\kappa=-\ds\frac{1+\nu-\sqrt{1+4\nu\sqrt{1-\nu^2}}}{\nu}$ and
$\rho=-\ds\frac{1+\nu-\sqrt{1+2\nu\sqrt{1-\nu^2}}}{\nu}$
are  functions of $\nu$ only. Thus we conclude that
there is a function $c=c(\nu)$ such that (\ref{FGR}) holds if and only if
$$
a''(C^2)\not=c(\nu)a'(C^2)/C^2,\quad\nu=a'(C^2)C^2/a(C^2).
$$
The function $c(\nu)$ is algebraic.
Hence, if the function $a(\cdot)$ is polynomial, or even
real analytic, then generically (\ref{FGR}) holds except possibly
at a discrete set of values of $C$. 

\end{document}